\documentclass[a4paper,10pt,oneside]{article}
\usepackage{comment}
\usepackage[utf8]{inputenc}
\usepackage{enumerate}
\usepackage{enumitem}
\usepackage[OT4]{fontenc}
\usepackage{graphicx}
\usepackage{amssymb}
\usepackage{amsmath}
\usepackage{amsfonts}
\usepackage{amsthm}
\usepackage{authblk}
\usepackage{wrapfig}
\usepackage[textsize=scriptsize]{todonotes}
\usepackage[boxruled,vlined,nokwfunc]{algorithm2e}
\usepackage{thmtools}
\usepackage{thm-restate}
\usepackage{microtype}
\usepackage{fullpage}
\usepackage{url}
\usepackage[hidelinks]{hyperref}

\newif\ifshort
\newif\iffull
\fulltrue

\newtheorem{theorem}{Theorem}[section]
\newtheorem{lemma}[theorem]{Lemma}

\newtheorem{fact}[theorem]{Fact}

\newtheorem{corollary}[theorem]{Corollary}
\newtheorem{example}[theorem]{Example}

\newcommand{\dual}[1]{\ensuremath{#1^{*}}}
\newcommand{\fv}[1]{\ensuremath{#1^{\diamond}}}

\newcommand{\contr}[2]{\ensuremath{#1 / #2}}
\newcommand{\remove}[2]{\ensuremath{#1 - #2}}

\newcommand{\eid}{\ensuremath{\mathrm{id}}}

\newcommand{\cvert}{\ensuremath{\phi}}

\newcommand{\poly}{\ensuremath{\mathrm{poly}}}
\newcommand{\bnd}[1]{\ensuremath{\partial #1}}
\newcommand{\rdiv}{\ensuremath{\mathcal{R}}}

\newcommand{\parf}{\ensuremath{\mathcal{T}}}

\newcommand{\doi}[1]{\textsc{doi}: \href{http://dx.doi.org/#1}{\nolinkurl{#1}}}

\renewcommand{\subparagraph}[1]{\paragraph{#1}}

\begin{document}

\title{Contracting a Planar Graph Efficiently}
%\title{Contracting a Planar Graph Efficiently\footnote{Giuseppe F. Italiano is partially supported by
%the Italian Ministry of Education, University and Research under
%Project AMANDA (Algorithmics for MAssive and Networked DAta).
%When working on this paper Jakub Łącki was partly supported by the EU FET project MULTIPLEX no. 317532 and the Google Focused Award on "Algorithms for Large-scale Data Analysis". Part of this work was done while Jakub Łącki was visiting the Simons Institute for the Theory of Computing.}}

\author[1]{Jacob Holm\thanks{This research is supported by the Advanced Grant DFF-0602-02499B from the Danish Council for Independent Research under the Sapere Aude research career programme.}}
\author[2]{Giuseppe F. Italiano\thanks{Partly supported by the Italian Ministry of Education, University and Research under Project AMANDA (Algorithmics for MAssive and Networked DAta).}}
\author[3]{Adam Karczmarz\thanks{Supported by the Polish National Science Center grant number 2014/13/B/ST6/01811.}}
\author[4]{\\Jakub Łącki\thanks{When working on this paper Jakub Łącki was partly supported by the EU FET project MULTIPLEX no. 317532 and the Google Focused Award on "Algorithms for Large-scale Data Analysis" and Polish National Science Center grant number 2014/13/B/ST6/01811. Part of this work was done while Jakub Łącki was visiting the Simons Institute for the Theory of Computing.}}
\author[1]{Eva Rotenberg}
\author[3]{Piotr Sankowski\protect\footnotemark[3]\thanks{The work of P. Sankowski is a part of the project TOTAL that has received funding from
the European Research Council (ERC) under the European Union’s Horizon 2020 research and innovation programme (grant agreement No 677651).}}

\affil[1]{University of Copenhagen, Denmark}
\affil[ ]{\texttt{\{jaho|roden\}@di.ku.dk}}

\affil[2]{University of Rome Tor Vergata}
\affil[ ]{\texttt{giuseppe.italiano@uniroma2.it}}

\affil[3]{University of Warsaw, Poland}
\affil[ ]{\texttt{\{a.karczmarz|sank\}@mimuw.edu.pl}}

\affil[4]{Google Research, New York}
\affil[ ]{\texttt{jlacki@google.com}}

\date{}

\maketitle

\begin{abstract}
We present a data structure that can maintain a simple planar graph under edge contractions in linear total time.
The data structure supports adjacency queries and provides access to neighbor lists in $O(1)$ time.
Moreover, it can report all the arising self-loops and parallel edges.

By applying the data structure, we can achieve optimal running times for decremental bridge detection, $2$-edge connectivity, maximal $3$-edge connected components, and the problem of finding a unique perfect matching for a static planar graph.
Furthermore, we improve the running times of algorithms for several planar graph problems, including decremental 2-vertex and 3-edge connectivity, and 
we show that using our data structure in a black-box manner, one obtains conceptually simple optimal algorithms for computing MST and 5-coloring in planar graphs.
\end{abstract}

\section{Introduction}
An edge contraction is one of the fundamental graph operations.
Given an undirected graph and an edge~$e$, contracting the edge $e$ consists in removing it from the graph and merging its endpoints.
The notion of a contraction has been used to describe a number of prominent graph algorithms, including Edmonds' algorithm for computing maximum matchings~\cite{Edmonds65paths} or Karger's minimum cut algorithm~\cite{Karger:1993}.

Edge contractions are of particular interest in planar graphs, as a number of planar graph properties are easiest described using contractions.
For example, it is well-known that a graph is planar precisely when it cannot be transformed into $K_5$ or $K_{3,3}$ by contracting edges or removing vertices or edges.
Moreover, contracting an edge preserves planarity.

While a contraction operation is conceptually very simple, its efficient implementation is challenging.
By using standard data structures (e.g. balanced binary trees), one can maintain adjacency lists of a graph in polylogarithmic amortized time.
However, in many planar graph algorithms this becomes a bottleneck.
As an example, 
consider the problem of computing a $5$-coloring of a planar graph. There exists a very simple algorithm based on contractions~\cite{Tarjan80}, but efficient implementations use some more involved planar graph properties~\cite{Frederickson84,Tarjan80,Robertson:1996}. For example, the algorithm by Matula, Shiloach and Tarjan~\cite{Tarjan80} uses the fact that every planar graph has either a vertex of degree at most 4 or a vertex of degree 5 adjacent to at least four vertices each having degree at most 11.
Similarly, although there exists a very simple algorithm for computing a MST of a planar graph based on edge contractions, various different methods have been used to implement it efficiently~\cite{Frederickson84,Mares,Matsui}.

\ifshort
\vspace{-2mm}
\fi
\subparagraph{Our Results.}
We show a data structure that can efficiently maintain a planar graph subject to edge contractions in $O(n)$ total time, 
assuming the standard word-RAM model with word size $\Omega(\log{n})$.
It can report groups of parallel edges and self-loops that emerge.
It also supports constant-time adjacency queries and maintains the neighbor lists and degrees explicitly.
The data structure can be used as a black-box to implement planar graph algorithms that use contractions.
In particular, it can be used to give clean and conceptually simple implementations of the algorithms for computing $5$-coloring or MST
that do not manipulate the embedding.
More importantly, by using our data structure we give improved algorithms for a few problems in planar graphs. In particular, we obtain optimal algorithms for decremental $2$-edge-connectivity, finding unique perfect matching, and computing maximal $3$-edge-connected subgraphs.
We also obtain improved algorithms for decremental $2$-vertex and $3$-edge connectivity,
where the bottleneck in the state-of-the-art algorithms \cite{Giammarresi:96} is detecting parallel edges under contractions.
For detailed theorem statements, see Sections~\ref{sec:interface} and~\ref{sec:overview}.

\ifshort
\vspace{-2mm}
\fi

\subparagraph{Related work.} The problem of detecting
self-loops and parallel edges under contractions is implicitly addressed by Giammarresi and Italiano~\cite{Giammarresi:96} in their work on decremental (edge-, vertex-) connectivity in planar graphs.
Their data structure uses $O(n\log^2{n})$ total time.

In their book, Klein and Mozes~\cite{Klein:book} show that there exists a data structure
maintaining a planar graph under edge contractions and deletions and
answering adjacency queries in $O(1)$ worst-case time.
The update time is $O(\log{n})$.
This result is based on the work of Brodal and Fagerberg \cite{Brodal:1999}, who showed
how to maintain a bounded outdegree orientation of a dynamic planar graph
so that edge insertions and deletions are supported in $O(\log{n})$ amortized time.

Gustedt \cite{Gustedt} showed an optimal solution to the union-find
problem, in the case when at any time, the actual subsets form
disjoint, connected subgraphs of a given planar graph $G$.
In other words, in this problem the allowed unions correspond
to the edges of a planar graph and the execution of a union operation
can be seen as a contraction of the respective edge.
\ifshort
\vspace{-2mm}
\fi

\subparagraph{Our Techniques.}

It is relatively easy to give a simple \emph{vertex merging data structure} for general graphs,
that would process any sequence of contractions in $O(m\log^2{n})$ total time and support the same queries as our data structure in $O(\log{n})$ time.
To this end, one can store the lists $N(v)$ of neighbors of individual vertices as balanced binary trees.
Upon a contraction of an edge $uv$, or a more general operation of merging
two (not necessarily adjacent)
vertices $u,v$, 
$N(u)$ and $N(v)$ are merged by inserting the
smaller set into the larger one (and detecting loops and parallel edges by the way, at
no additional cost).
If we used hash tables instead of balanced BSTs, we could achieve $O(\log{n})$ expected
amortized update time and $O(1)$ query time.
In fact, such an approach was used 
in~\cite{Giammarresi:96}.

To obtain the speed-up we take advantage of planarity.
Our general idea is to partition the graph into small pieces
and use the above simple-minded vertex merging data structures  
to solve our problem separately for each of the pieces and
for the subgraph induced by the vertices contained in multiple pieces
(the so-called boundary vertices).
Due to the nature of edge contractions,
we need to specify how
the partition evolves when our graph changes.

The data structure builds an $r$-division (see Section~\ref{sec:preliminaries}) $\rdiv=P_1,P_2,\ldots$ of $G_0$ for $r=\log^4{n}$.
The set $\bnd{\rdiv}$ of boundary vertices (i.e., those shared among at least
two pieces) has size $O(n / \log^2 n)$.
Let $(V_0,E_0)$ denote the original graph, and $(V,E)$ denote the current graph (after performing some number of contractions).
Then we can denote by $\cvert:V_0\to V$ a function
such that the initial vertex $v_0\in V_0$ is contracted into $\cvert(v_0)$.
We use vertex merging data structures to detect parallel edges and self-loops in the ``top-level'' subgraph
$G[\cvert(\bnd{\rdiv})]$, which contains only edges between boundary vertices, and separately for the ``bottom-level'' subgraphs $G[\cvert(V(P_i))]\setminus G[\cvert(\rdiv)]$.
 At any time, each edge of $G$ is contained in exactly one
of the defined subgraphs, and 
thus, the distribution of responsibility for handling individual edges
is based solely on the initial $r$-division.

However, such an assignment of responsibilities gives rise to additional
difficulties.
First, a contraction of an edge in a lower-level subgraph might
cause some edges ``flow'' from this subgraph to the top-level subgraph (i.e., we may get new edges connecting boundary vertices).
As such an operation turns out to be costly in our implementation,
we need to prove that the number of such events is only 
$O(n/\log^2{n})$.

Another difficulty lies in the need of keeping the individual
data structures synchronized: when an edge of the top-level subgraph
is contracted, pairs of vertices in multiple lower-level subgraphs might
need to be merged.
We cannot afford iterating through all the lower-level subgraphs
after each contraction in $G[\cvert(\bnd{\rdiv})]$.
This problem is solved by maintaining a system of pointers between representations of the
same vertex of $V$ in different data structures
and another clever application of the smaller-to-larger merge strategy.

Such a two-level data structure would yield a data structure with
$O(n\log\log{n})$ total update time.
To obtain a linear time data structure, we further partition the pieces $P_i$ and
add another layer of maintained subgraphs on $O(\log^4 \log^4 n) = O(\log^4 \log n)$ vertices.
These subgraphs are so small that we can precompute in $O(n)$ time the self-loops and parallel edges
for every possible graph on $t$ vertices and every possible sequence of edge contractions.

We note that this overall idea of recursively reducing a problem with an $r$-division to a size when microencoding
can be used has been previously exploited in~\cite{Gustedt}~and~\cite{Lacki:2015} (Gustedt~\cite{Gustedt} did not use $r$-divisions, but his concept of a \emph{patching} could be replaced with an $r$-division).
Our data structure can be also seen as a solution to a more general version of the planar union-find
problem studied by Gustedt~\cite{Gustedt}.
However, maintaining the status of each edge~$e$ of the initial graph $G$ (i.e., whether
$e$ has become a self-loop or a parallel edge) subject to edge contractions
turns out to be a serious technical challenge.
For example, in~\cite{Gustedt}, the requirements posed on the bottom-level
union-find data structures are in a sense relaxed and it is not necessary for those
to be synchronized with the top-level union-find data structure.
\ifshort
\vspace{-3mm}
\fi
\subparagraph{Organization of the Paper.}
The remaining part of this paper is organized as follows.
In Section~\ref{sec:preliminaries}, we introduce the needed notation and definitions,
whereas in Section~\ref{sec:interface} we define the operations that our data structure
supports.
Then, in Section~\ref{sec:overview} we present a series of applications of our data structure.
In Section~\ref{sec:loopdetection}, we provide a detailed implementation of our data structure.
\ifshort
Due to space constraints, many of the proofs, along with
the pseudocode for example algorithms using our data structure,
can be found in the full version of this paper.
\fi

\section{Preliminaries}\label{sec:preliminaries}
Throughout the paper we use the term \emph{graph} to denote an undirected \emph{multigraph}, that is we allow the graphs to have parallel edges and self-loops.
Formally, each edge $e$ of such a graph is a pair $(\{u,w\},\eid(e))$ consisting of a pair of vertices and a unique identifier used to distinguish between the parallel edges.
For simplicity, we skip this third coordinate and use just $uw$ to denote one of the edges connecting vertices $u$ and $w$.
If the graph contains no parallel edges and no self-loops, we call it \emph{simple}.

For any graph $G$, we denote by $V(G)$ and $E(G)$ the sets of vertices and edges
of $G$, respectively.
A graph $G'$ is called a subgraph of $G$ if $V(G')\subseteq V(G)$ and $E(G')\subseteq E(G)$.
We define $G_1\cup G_2=(V(G_1)\cup V(G_2),E(G_1)\cup E(G_2))$ and
$G_1\setminus G_2=(V(G_1),E(G_1)\setminus E(G_2))$.
For $S\subseteq V(G)$, we denote by $G[S]$ the \emph{induced subgraph} $(S,\{uv: uv\in E(G), \{u,v\}\subseteq S\})$.

For a vertex $v\in V$, we define $N(v)=\{u:uv\in E, u\neq v\}$ to be the \emph{neighbor set} of $v$.

A \emph{cycle} of a graph $G$ is a nonempty set $C \subseteq E(G)$, such that for some ordering of edges $C = \{u_1w_1, \ldots, u_kw_k\}$, we have $w_i = u_{i+1}$ for $1 \leq i < k$ and $w_k = u_1$, and the vertices $u_1,\ldots,u_k$ are distinct.
The \emph{length} of a cycle $C$ is simply $|C|$.
Note that this definition allows cycles of length $1$ (self-loop) or $2$ (a pair of parallel edges),
but does not allow non-simple cycles of length $3$ or more.
A \emph{cut} is a minimal (w.r.t. inclusion) set $C \subseteq E(G)$, such that $G \setminus C$ has more connected components than $G$.

Let $G=(V,E)$ be a graph and $xy=e \in E$.
We use $\remove{G}{e}$ to denote the graph obtained from $G$ by removing $e$
and $\contr{G}{e}$ to denote the graph obtained by contracting an edge $e$
(in the case of a contraction $e$ may not be a self-loop, i.e., $x\neq y$).
We will often look at contraction from the following perspective:
as a result of contracting $e$, all edge endpoints equal to $x$ or $y$ are replaced
with some new vertex $z$.
In some cases it is convenient to assume $z\in \{x,y\}$.
This yields a 1-to-1 correspondence between the edges of $\remove{G}{e}$ and the edges of $\contr{G}{e}$.
Formally, we assume that the contraction preserves the edge identifiers,
i.e., $e_1\in E(\remove{G}{e})$ and $e_2\in E(\contr{G}{e})$ are corresponding
if and only if $\eid(e_1)=\eid(e_2)$.

Note that contracting an edge may introduce parallel edges and self-loops.
Namely, for each edge that is parallel to $e$ in $G$, there is a self-loop in $\contr{G}{e}$. And for each cycle of length $3$ that contains $e$ in $G$, there is a pair of parallel edges in $G/e$.
\ifshort
\vspace{-2mm}
\fi
\subparagraph{Planar graphs.}
An embedding of a planar graph is a mapping of its vertices to distinct points
and edges to non-crossing curves in the plane.
We say that a planar graph $G$ is \emph{plane},
if some embedding of $G$ is assumed.
A face of a connected plane $G$ is a maximal open connected set of points
not in the image of any vertex or edge in the embedding of~$G$.
\ifshort
\vspace{-2mm}
\fi
\iffull
\subparagraph{Semi-strictness.}
\begin{wrapfigure}[6]{r}{0.37\textwidth}
    \vspace{-2em}
    \includegraphics[width=0.34\textwidth]{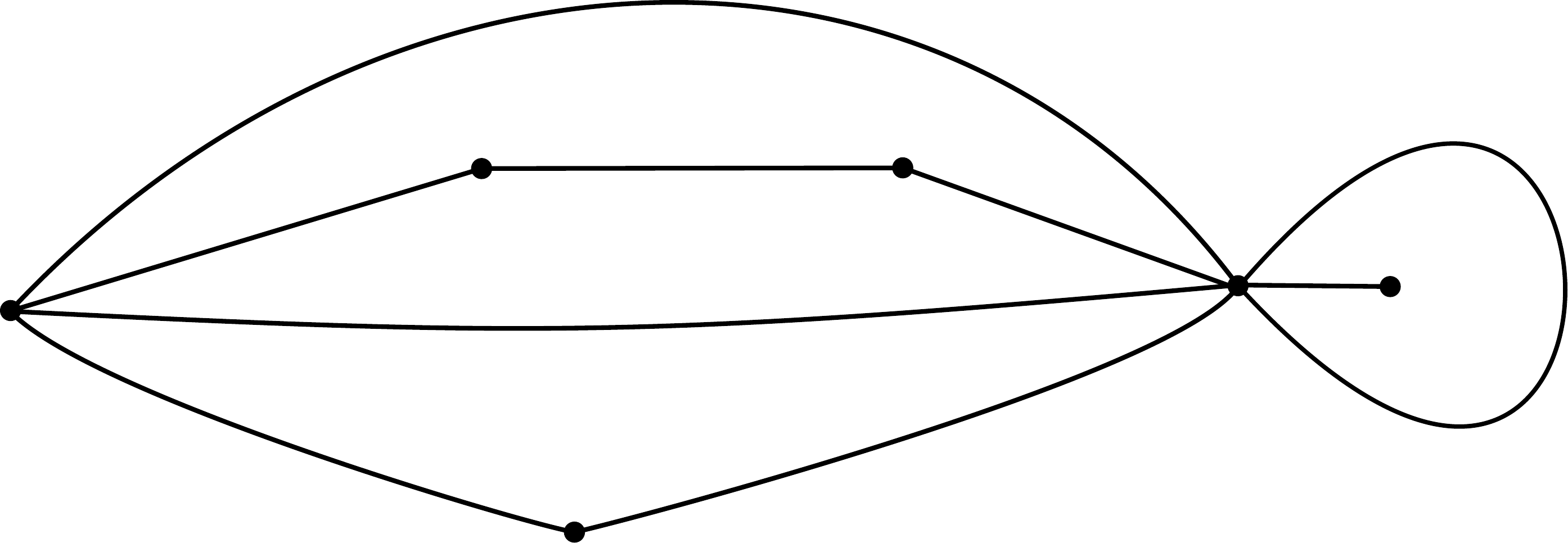}
    \caption{A semi-strict graph with $6$ vertices and $5$ faces.\label{fig:quasi}}
\end{wrapfigure}
We say that a connected plane graph $G$ is \emph{semi-strict} \cite{Klein:book} if each
of its faces has a boundary of at least $3$ edges (see Figure~\ref{fig:quasi}).
We can obtain a \emph{maximal semi-strict subgraph} of a plane embedded multigraph $H$ as follows:
for each set $P$ of parallel edges $xy$ of $H$ such that they form a contiguous
fragment of the edge rings of both $x$ and $y$, remove from $H$ all edges of $P$ except one.

\begin{fact}\label{f:factor3}
  A semi-strict plane graph $G$ with $n$ vertices has at most $3n-6$ edges.
\end{fact}
\begin{proof}
We note that each face of $G$ has at least $3$ edges and apply the Euler's formula.
\end{proof}
\fi
\subparagraph{Duality.}
Let $G$ be a plane graph. We denote by $\dual{G}$ the dual graph of $G$.
Each edge of $G$ naturally corresponds to an edge of $\dual{G}$.
We denote by $\dual{e}$ the edge of $\dual{G}$ that corresponds to $e \in E(G)$.
More generally, if $E_1 \subseteq E(G)$ is a set of edges of $G$, we set $\dual{E_1}=\{ \dual{e} | e \in E_1\}$. 

We exploit the following relations between $G$ and $\dual{G}$.
Deleting an edge $e$ of $G$ corresponds to contracting the edge $\dual{e}$ in $\dual{G}$, that is $\dual{(\remove{G}{e})} = \contr{\dual{G}}{\dual{e}}$.
Moreover, $C\subseteq E$ is a cut in $G$ iff $\dual{C}$ is a cycle in $\dual{G}$.
In particular, a bridge $e$ in $G$ corresponds to a self-loop in $\dual{G}$ and a two-edge cut in $G$ corresponds to a pair of parallel edges in $\dual{G}$.
\ifshort
\vspace{-2mm}
\fi
\subparagraph{Planar graph partitions.}
Let $G$ be a simple planar graph.
Let a \emph{piece} be subgraph of $G$ with no isolated vertices.
For a piece $P$, we denote by $\bnd{P}$ the set of vertices $v\in V(P)$
such that $v$ is adjacent to some edge of $G$ that is not contained in $P$.
$\bnd{P}$ is also called the set of \emph{boundary vertices} of $P$.
An $r$-division $\rdiv$ of $G$ is a partition of $G$ into $O(n/r)$
edge-disjoint pieces such that each piece $P\in\rdiv$ has $O(r)$ vertices
and $O(\sqrt{r})$ boundary vertices.
For an $r$-division $\rdiv$, we also denote by $\bnd{\rdiv}$ the set $\bigcup_{P_i\in \rdiv}\bnd{P_i}$.
Clearly, $|\bnd{\rdiv}|=O(n/\sqrt{r})$.
\ifshort
\vspace{-2mm}
\fi
\begin{lemma}[\cite{Goodrich:95, Klein:13, Walderveen:2013}]\label{lem:rdiv}
An $r$-division of a planar graph $G$ can be computed in linear time.
\end{lemma}
\ifshort
\vspace{-2mm}
\fi

\vspace{-1mm}
\section{The Data Structure Interface}\label{sec:interface}

In this section we specify the set of operations that our data
structure supports so that it fits our applications.
It proves beneficial
to look at the graph undergoing contractions
from two perspectives.
\begin{enumerate}
\item The \emph{adjacency viewpoint} allows us to track the neighbor sets
of the individual vertices, as if $G$ was simple at all times.
\item The \emph{edge status viewpoint} allows us to track, for all the original
  edges $E_0$, whether they became self-loops or parallel edges, and also track how
  $E_0$ is partitioned into classes of pairwise-parallel edges.
\end{enumerate}

Let $G_0=(V_0,E_0)$ be a planar graph used to initialize the data structure.
Recall that any contraction alters both the set of vertices and the set of edges of the graph.
Throughout, we let $G=(V,E)$ denote the \emph{current} version of the graph, unless otherwise stated.

Each edge $e\in E(G)$ can be either a self-loop, an edge parallel to some
other edge $e'\neq e$ (we call such an edge \emph{parallel}), or an edge
that is not parallel to any other edge of $G$ (we call it \emph{simple} in this case).
An edge $e\in E(G)$ that is simple might either get contracted
or might change into a parallel edge as a result of contracting other edges.
Similarly, a parallel edge might either get contracted or might
change into a self-loop.
Note that, during contractions, neither can a parallel edge ever become simple, nor can a self-loop become parallel.

Observe that parallelism is an equivalence relation
on the edges of $G$.
Once two edges $e_1,e_2$ connecting vertices $u,v\in V$ become parallel, they stay parallel until some edge $e_3$ (possibly equal to $e_1$ or $e_2$)
parallel to both of them gets contracted.
However, groups of parallel edges might merge (Figure~\ref{fig:bundles})
and this might also be a valuable piece of information.

\ifshort
\begin{wrapfigure}[7]{r}{0.3\textwidth}
    \vspace{-0.7em}
\fi
\iffull
\begin{wrapfigure}[9]{r}{0.3\textwidth}
\fi
    \includegraphics[width=0.27\textwidth]{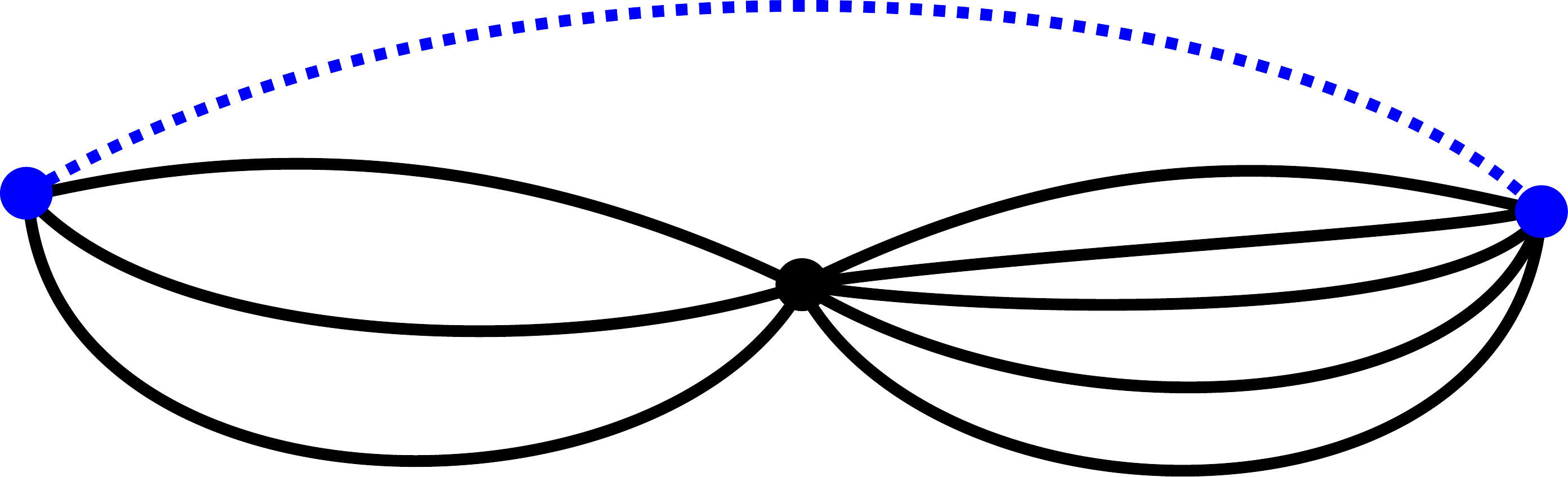}
    \caption{Contracting the blue dotted edge will merge two groups of parallel edges.\label{fig:bundles}}
\end{wrapfigure}
To succinctly describe how the groups of parallel edges change, we report
parallelism in a directed manner, as follows.
Each group $Y\subseteq E$ of parallel edges in $G$ is assumed to have its
\emph{representative} edge $\alpha(Y)$.
For $e \in Y$ we define $\alpha(e) = \alpha(Y)$.
When two groups of parallel edges $Y_1,Y_2\subseteq E$ merge as a result
of a contraction, the data structure chooses $\alpha(Y_i)$ for some $i\in\{1,2\}$
to be the new representative of the group $Y_1\cup Y_2$ and reports
an ordered pair $\alpha(Y_{3-i})\to \alpha(Y_i)$ to the user.
We call each such pair a \emph{directed parallelism}.
After such an event, $\alpha(Y_{3-i})$ will not be reported as a part
of a directed parallelism anymore.
The choice of $i$
can also be made according to some fixed strategy, e.g., if the edges are assigned weights $\ell(\cdot)$ then
we may choose $\alpha(Y_i)$ so that $\ell(\alpha(Y_i))\leq\ell(\alpha(Y_{3-i}))$.
This is convenient in what Klein and Mozes \cite{Klein:book} call \emph{strict optimization problems},
such as MST,
where we can discard one of any two parallel edges based only on these edges.

Note that at any point of time the set of directed parallelisms reported so far
can be seen as a forest of rooted trees $\parf$, such that each tree $T$ of $\parf$ represents
a group $Y$ of parallel edges of $G$. The root of $T$ is equal to $\alpha(Y)$.

When some edge is contracted, all edges parallel to it are reported as self-loops.
Clearly, each edge $e$ is reported as a self-loop at most once.
Moreover, it is reported as a part of a directed parallelism $e\to e'$, $e'\neq e$, at most
once.

We are now ready to define the complete interface of our data structure.
\newcommand{\dsinit}{\ensuremath{\textup{\texttt{init}}}}
\newcommand{\dscontr}{\ensuremath{\textup{\texttt{contract}}}}
\newcommand{\dsedge}{\ensuremath{\textup{\texttt{edge}}}}
\newcommand{\dsverts}{\ensuremath{\textup{\texttt{vertices}}}}
\newcommand{\dsdeg}{\ensuremath{\textup{\texttt{deg}}}}
\newcommand{\dsnei}{\ensuremath{\textup{\texttt{neighbors}}}}
\newcommand{\nil}{\ensuremath{\textbf{nil}}}
\begin{itemize}
  \item $\dsinit(G_0=(V_0,E_0),\ell)$: initialize the data structure. $\ell$ is an optional weight function.
  \item $(s,P,L):=\dscontr(e)$, for $e\in E$:
    contract the edge $e$.
    Let $e=uv$.
    The call $\dscontr(e)$ returns a vertex $s$ resulting from
    merging $u$ and $v$, and two lists $P$, $L$ of new directed parallelisms
    and self-loops, respectively, reported as a result of contraction of $e$.
  \item $\dsverts(e)$, for $e\in E$: return $u,v\in V$ such that $e=uv$.
  \item $\dsnei(u)$, for $u\in V$: return an iterator to the list $\{(v,\alpha(uv)):v\in N(u)\}$.
  \item $\dsdeg(u)$, for $u\in V$: find the number of neighbors of $u$ in $G$.
  \item $\dsedge(u,v)$, for $u,v\in V$: if $uv\in E$, then return $\alpha(uv)$.
    Otherwise, return $\nil$.
\end{itemize}

The following theorem summarizes the performance of
our data structure.
\begin{theorem}\label{thm:main}
    Let $G=(V,E)$ be a planar graph with $|V|=n$ and $|E|=m$.
    There exists a data structure supporting
 $\dsedge$, $\dsverts$, $\dsnei$ and $\dsdeg$
    in $O(1)$ worst-case time, and whose initialization and any sequence of $\dscontr$ operations take $O(n+m)$ expected time, or $O(n+m)$ worst-case time, if no $\dsedge$ operations
    are performed.
    The data structure supports iterating through the neighbor list of a vertex
    with $O(1)$ overhead per element.
\end{theorem}

\section{Applications}\label{sec:overview}

\subparagraph{Decremental Edge- and Vertex-Connectivity.}
In the \emph{decremental $k$-edge ($k$-vertex) connectivity} problem, the goal is to design a data
    structure that supports queries about the existence of $k$ edge-disjoint (vertex-disjoint)
    paths between a pair of given vertices, subject to edge deletions.
    We obtain improved algorithms for decremental 2-edge-, 2-vertex- and 3-edge-connectivity
    in dynamic planar graphs.
    For decremental 2-edge-connectivity we obtain an optimal data structure with both
    updates and queries supported in amortized $O(1)$ time.
    In the case of 2-vertex- and 3-edge-connectivity, we achieve the amortized update time
    of $O(\log{n})$, whereas the query time is constant.
    For all these problems, we improve upon the 20-year-old update bounds by Giammarresi and Italiano
    \cite{Giammarresi:96} by a factor of $O(\log{n})$.

\begin{theorem}\label{thm:2edge_ds}
  Let $G=(V,E)$ be a planar graph and let $n=|V|$.
  There exists a deterministic data structure that maintains $G$ subject to edge deletions and can answer $2$-edge connectivity queries in $O(1)$ time.
  Its total update time is $O(n)$.
\end{theorem}
\begin{proof}
  Denote by $G_0$ the initial graph.
Suppose wlog. that $G_0$ is connected.
  Let $B(G)$ be the set of all bridges of $G$.
  Note that two vertices $u,v$ are in the same
  2-edge-connected component of $G$ iff they
  are in the same connected component of the graph $(V,E\setminus B(G))$.

  Observe that if $e$ is a bridge, then deleting $e$ from $G$ does not
  influence the 2-edge-components of $G$.
  Hence, when a bridge $e$ is deleted, we may ignore this deletion.
  We denote by $G'$ be the graph obtained from $G_0$ by the same sequence of deletions
  as $G$, but ignoring the bridge deletions.
  This way, $G'$ is connected at all times and the 2-edge-connected components
  of $G'$ and $G$ are the same.
  It is also easy to see that $E(G)\setminus B(G)=E(G')\setminus B(G')$
  and $B(G)=B(G')\cap E(G)$.
  Moreover, the set $E(G')$ shrinks in time whereas $B(G')$ only grows.

  First we show how the set $B(G')$ is maintained.
  Recall that $e\in E(G')$ is a bridge of $G'$~iff
  $\dual{e}$ is a self-loop of $\dual{G'}$.
  We build the data structure of Theorem~\ref{thm:main}
  for $\dual{G'}$, which initially equals $\dual{G_0}$.
  As deleting a non-bridge edge $e$ of $G'$ translates to a contraction
  of a non-loop edge $\dual{e}$ in $\dual{G'}$, we can maintain
  $B(G')$ in $O(n)$ total time by detecting self-loops in $\dual{G'}$.

  Denote by $H$ the graph $(V,E(G')\setminus B(G'))$.
  To support 2-edge connectivity queries, 
  we maintain the graph $H$ with the decremental connectivity data structure
  of Łącki and Sankowski \cite{Lacki:2015}. This data structure maintains
  a planar graph subject to edge deletions in linear total time and
  supports connectivity queries in $O(1)$ time.
  When an edge $e$ is deleted from $G$, we first check whether it
  is a bridge and if so, we do nothing.
  If $e$ is not a bridge, the set $E(G')$ shrinks and thus we remove the edge $e$ from $H$.
  The deletion of $e$ might cause the set $B(G')$ to grow.
  Any new edge of $B(G')$ is also removed from $H$ afterwards.

  To conclude, note that each 2-edge connectivity query on $G$ translates
  to a single connectivity query in $H$.
  All the maintained data structures have $O(n)$ total update~time.
\end{proof}

As an almost immediate consequence of Theorem~\ref{thm:2edge_ds} we improve upon \cite{Gabow:2001}
and obtain an optimal
algorithm for the \emph{unique perfect matching} problem when restricted to planar graphs.
\iffull
The details can be found in Appendix~\ref{a:omitted}.
\fi
\ifshort
\vspace{-2mm}
\fi
\iffull
\begin{restatable}{cor}{uniquematching}
  Given a planar graph $G=(V,E)$ with $n=|V|$, in $O(n)$ time we can find a unique perfect matching
of $G$ or detect that the number of perfect matchings in $G$ is not $1$.
\end{restatable}
\fi
\ifshort
\begin{corollary}
  Given a planar graph $G=(V,E)$ with $n=|V|$, in $O(n)$ time we can find a unique perfect matching
of $G$ or detect that the number of perfect matchings in $G$ is not $1$.
\end{corollary}
\fi
To obtain improved bounds for $2$-vertex connectivity and $3$-edge connectivity we use the data structure of Theorem~\ref{thm:main} to remove bottlenecks in the existing algorithms by Giammarresi and Italiano~\cite{Giammarresi:96}.
\iffull
The details are deferred to Appendix~\ref{a:omitted}.
\fi
\begin{theorem}\label{thm:2vert3edge}
  Let $G=(V,E)$ be a planar graph and let $n=|V|$.
  There exists a deterministic data structure that maintains $G$ subject to edge deletions and can answer
  $2$-vertex connectivity and $3$-edge connectivity queries in $O(1)$ time.
  Its total update time is $O(n \log n)$.
\end{theorem}
\ifshort
\vspace{-2mm}
\fi
\subparagraph{Maximal 3-Edge-Connected Subgraphs.} 
A $k$-edge-connected component of a graph $G$ is a maximal (w.r.t. inclusion) subset $S$ of vertices, such that each pair of vertices in $S$ is $k$-edge-connected.
However, if $k \geq 3$, in the subgraph of $G$ induced by $S$, some pairs of vertices may not be $k$-edge-connected (see~\cite{Chechik:2017} for an example).
Thus, for $k \geq 3$, maximal $k$-edge-connected subgraphs can be different from $k$-edge-connected components.
Very recently, Chechik et al.~\cite{Chechik:2017} showed how to compute maximal $k$-edge-connected subgraphs %of a graph
 in $O((m+n\log n)\sqrt{n}\,)$ time
for any constant $k$, or $O(m\sqrt{n}\,)$ time for $k=3$.
Using the results of~\cite{Giammarresi:96} one can compute maximal $3$-edge-connected subgraphs of a planar multigraph in $O(m+n \log n)$ time. Our new approach allows us to improve this to an optimal $O(m+n)$ time bound.

\iffull
\begin{restatable}{lem}{maxedgecon}
The maximal $3$-edge-connected subgraphs of a planar graph can be
computed in linear time.
\end{restatable}
\fi
\ifshort
\begin{lemma}
The maximal $3$-edge-connected subgraphs of a planar graph can be
computed in linear time.
\end{lemma}
\fi
\ifshort
\vspace{-4mm}
\fi
\iffull
\newpage
\fi
\subparagraph{Simple Linear-Time Algorithms.} Finally, we present two examples showing
that Theorem~\ref{thm:main} might be a useful
black-box in designing linear time algorithms for planar graphs.
\ifshort
The details and the relevant pseudocode can be found in the full version of this paper.
\fi
\iffull
\begin{wrapfigure}[10]{r}{0.4\textwidth}
\vspace{0.5em}
   \includegraphics[width=0.25\textwidth]{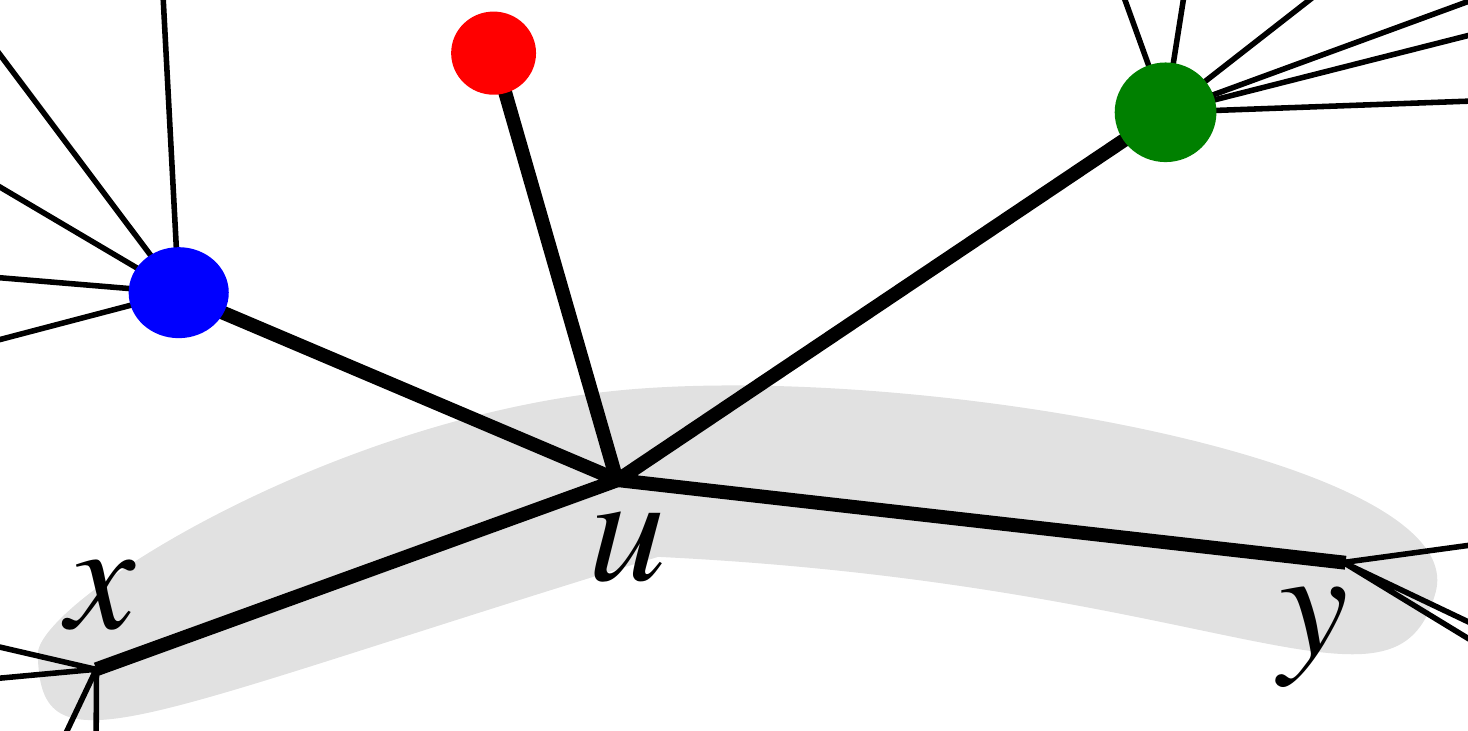}
    \caption{The degree $\le 5$ vertex 
        and its two independent neighbors may be colored using the remaining two colors.
        \label{fig:5col}
    }
\end{wrapfigure}
\fi
\ifshort
\begin{wrapfigure}[9]{r}{0.4\textwidth}
    \includegraphics[width=0.25\textwidth]{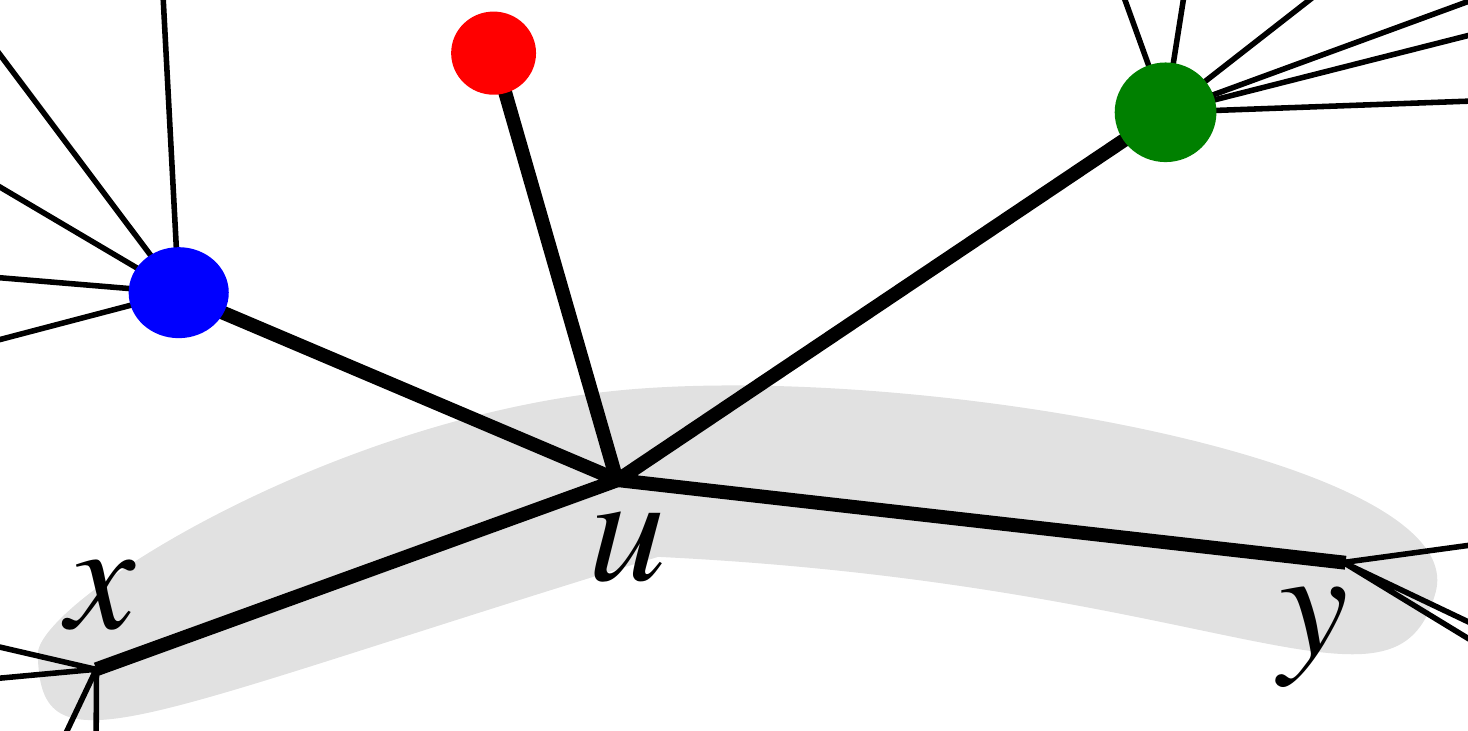}
    \caption{The degree $\le 5$ vertex 
        and its two independent neighbors may be colored using the remaining two colors.
        \label{fig:5col}
    }
\end{wrapfigure}
\fi
\ifshort
\vspace{-2mm}
\fi
\begin{example}\label{ex:5col}
Every planar graph $G$ can be $5$-colored in expected linear time.
\end{example}
\iffull
\fi

\begin{proof}
A textbook proof of the $5$-color theorem proceeds by induction as follows (see Figure~\ref{fig:5col}).
Each simple planar graph has a vertex $u$ of degree at most~$5$.
  The case when $u$ has degree less than $5$ is easy: for any $v\in N(u)$,
  we can color $\contr{G}{uv}$ inductively,
  uncontract the edge $uv$ and finally recolor $u$ with a color
  not used among the vertices $N(u)$.
When, however, $u$ has degree exactly $5$, there exist two neighbors
of $x,y$ of $u$ such that $x$ and $y$ are not adjacent, as otherwise
$G$ would contain $K_5$.
We could thus obtain a planar graph $G'$ by contracting both $ux$ and $uy$.
  After inductively coloring $G'$ and ``uncontracting'' $ux$ and $uy$,
  we obtain a coloring of $G$ that is valid, except that $x$, $y$ and $u$ have the same
  colors assigned.
  Thus, at most $4$ colors are used among the neighbors of $u$ and we 
  recolor $u$ to the remaining color in order to get a valid coloring of $G$.

Note that this proof can be almost literally converted into a linear time $5$-coloring
  algorithm
  \ifshort
  (see the full version for the pseudocode)
  \fi
  \iffull
  (see Appendix~\ref{a:pseudocode} for the pseudocode)
  \fi
  using the data structure of Theorem~\ref{thm:main} built for $G$.
  We only need to maintain a subset $Q$ of vertices of $G$ with degree at most~$5$.
  The subset $Q$ can be easily maintained in linear total time,
  since all vertices that potentially
  change their degrees after the call $\texttt{contract}(e)$
  are endpoints of the reported parallel edges.
\end{proof}

\ifshort
\vspace{-2mm}
\fi
\begin{example}
An MST of a planar graph $G$ can be computed in linear time.
\end{example}
\iffull
\begin{proof}
  Observe that by the cut property of a minimum spanning tree, for any vertex $u\in V(G)$,
  and an edge $e$ of minimum
  cost among the edges adjacent to $u$, there exists a minimum spanning tree $T$ of $G$
  such that $e\in T$.
  
  This observation can be turned into an efficient algorithm as follows.
  Again we build the data structure of Theorem~\ref{thm:main}
  and maintain the subset $S\subseteq V(G)$ containing the vertices
  of degree no more than $5$.
  We repeatedly pick any $u\in S$, find the minimum cost edge $uv$ adjacent
  to $u$ (in $O(1)$ time), include $uv$ in the constructed MST and subsequently
  contract $uv$.
  The set $S$ can be updated after a contraction analogously as in Example~\ref{ex:5col}.
  The pseudocode can be found in Appendix~\ref{a:pseudocode}.
  By Theorem~\ref{thm:main}, the total running time of this algorithm is linear.
\end{proof}

\fi

\section{Maintaining a Planar Graph Under Contractions}\label{sec:loopdetection}
In this section we prove Theorem~\ref{thm:main}.
\ifshort
We defer the discussion on supporting arbitrary weights $\ell(\cdot)$ to the full version.
\fi
\iffull
We defer the discussion on supporting arbitrary weights $\ell(\cdot)$ to Appendix~\ref{sec:edgeweights}.
\fi
Hence, in the following, we assume all edges have equal weights.
\ifshort
\vspace{-3mm}
\fi
\subsection{A Vertex Merging Data Structure}\label{sec:vids}
We first consider a more general problem,
which we call the \emph{bordered vertex merging} problem.
The data structure presented below will constitute a basic
building block of the multi-level data structure.
Let us now describe the data structure for the bordered vertex merging problem in detail.
Suppose we have a dynamic \emph{simple} planar graph $G=(V,E)$ and a \emph{border set} $B\subseteq V$.
Assume $G$ is initially equal to $G_0=(V_0,E_0)$ and no edge of $E_0$ connects two vertices of $B$.
The data structure handles the following update operations.
\begin{itemize}
  \item Merge (or in other words, an identification) of two vertices $u,v\in V$ ($u\neq v$), such that the graph is still planar.
    If $\{u,v\}\not\subseteq B$, then $u$ and $v$ have to be connected by an edge
    and in such a case the merge is equivalent to a contraction of $uv$.
  \item Insertion of an edge $e=uv$ (where $uv\notin E$ is not required), preserving planarity.
\end{itemize}

After each update operation the data structure reports the parallel edges and self-loops that emerge. Once reported, each set of parallel edges is merged into one representative edge.
Moreover, the data structure reports and removes any edges that have both endpoints in $B$.
Thus, the following invariants are satisfied before the first and after each modification:
\begin{enumerate}
\item $G$ is planar and simple.
\item No edge of $E$ has both its endpoints in $B$.
\end{enumerate}

Clearly, merging vertices alters the set $V$ by replacing
two vertices $u,v$ with a single vertex.
Thus, at each step, each vertex of $G$ corresponds to a set of vertices of the initial graph $G_0$.
We explicitly maintain a mapping $\cvert:V_0\to V$ such that for $a\in V_0$, $\cvert(a)$
is a vertex of the current vertex set $V$ ``containing'' $a$.
The reverse mapping $\cvert^{-1}:V\to 2^{V_0}$ is also stored explicitly.
We now define how the merge of $u$ and $v$ influences the set $B$.
When $\{u,v\}\subseteq B$, the resulting vertex is also
in $B$.
When $u\in B, v\notin B$ (or $v\in B,u\notin B$, resp.),
the resulting vertex is included in $B$ in place of $u$ ($v$, resp.).
Finally, for $u,v\notin B$,
the resulting vertex does not belong to $B$ either.

Let $\tilde{E}$ be the set of inserted edges.
At any time, the edges of $E$ constitute a subset of $E_0\cup\tilde{E}$ in the following sense:
for each $e=xy\in E$ there exists an edge $e'=uv\in E_0\cup\tilde{E}$ such that $\eid(e)=\eid(e')$,
and vertices $u$ and $v$ have been merged into $x$ and $y$, respectively.
%$\cvert^{-1}(u)\subseteq \cvert^{-1}(x)$ and $\cvert^{-1}(v)\subseteq\cvert^{-1}(y)$.

Note that some modifications might break the second invariant:
both an edge insertion and a merge might introduce an edge $e$ with both endpoints
in $B$.
We call such an edge a \emph{border edge}.
Each border edge $e$ that is not a self-loop
is reported and deleted from (or not inserted to) $G$.
Apart from reporting and removing new edges of $B\times B$ appearing in $E$, we also report
the newly created parallel edges that might arise
after the modification and remove them.
The reporting of parallel edges is done in the form of directed parallelisms,
as described in Section~\ref{sec:interface}.
Again, it is easy to see that each edge of $E_0\cup \tilde{E}$ is reported
as the first coordinate of a directed parallelism at most once.

Note that an edge $e$ may be first reported parallel (in a directed parallelism
of the form $e'\to e$, where $e'\neq e$) and then reported border.
\ifshort
\vspace{-3mm}
\fi
\subparagraph{The Graph Representation.}
The data structure for the bordered vertex merging problem internally maintains $G$ using the data structure of the following lemma for planar graphs.
\begin{lemma}[\cite{Brodal:1999}]\label{lem:adjq}
  There exists a deterministic, linear-space data structure, initialized in $O(n)$
  time, and maintaining a
  dynamic, simple planar graph $H$ with $n$ vertices, so that:
  \begin{itemize}
  \item adjacency queries in $H$ can be performed in $O(1)$ worst-case time,
  \item edge insertions and deletions can be performed in $O(\log{n})$ amortized time.
  \end{itemize}
\end{lemma}

\begin{fact}\label{fac:adj}
  The data structure of Lemma~\ref{lem:adjq} can be easily extended so that:
\begin{itemize}
  \item Doubly-linked lists $N(v)$ of neighbors, for $v\in V$, are maintained within the same bounds.
\item For each edge $xy$ of $H$, some auxiliary data associated with $e$ can
  be accessed and updated in $O(1)$ worst-case time.
\end{itemize}
\end{fact}

In addition to the data structure of Lemma~\ref{lem:adjq} representing $G$,
for each unordered pair $x,y$ of vertices adjacent in $G$,
we maintain an edge $\alpha(x,y)=e$,
where $e$ is the unique edge in $E$ connecting $x$ and $y$.
Recall that in fact $\alpha(x,y)$ corresponds to some of the original
edges of $E_0$ or one of the inserted edges $\tilde{E}$.
By Fact~\ref{fac:adj}, we can access $\alpha(x,y)$ in constant time.

The mapping $\cvert$ is stored in an array, whereas the sets $\cvert^{-1}(\cdot)$ -- in doubly-linked lists.

Suppose we merge two vertices $u,v\in V$.
Instead of creating a new vertex $w$, we merge one of these
vertices into the other.
Suppose we merge $u$ into $v$.
In terms of the operations supported by the data structure of Lemma~\ref{lem:adjq},
we need to remove each edge $ux$
and insert an edge $vx$, unless $v$ has been adjacent to $x$ before.

To update our representation, we only need to perform the following steps:
\begin{itemize}
  \item For each $v_0\in \cvert^{-1}(u)$, set $\cvert(v_0)=v$ and add $v_0$ to $\cvert^{-1}(v)$.
  \item Compute the list $N_u = \{(x,\alpha(u,x)) : x\in N(u)\}$.
    Remove all edges adjacent to $u$ from~$G$.
    For each $(x,\alpha(u,x))\in N_u$, $x\neq v$, check whether $x\in N(v)$
    (this can be done in $O(1)$ time, by Lemma~\ref{lem:adjq}).
    If so, report the parallelism $\alpha(u,x)\to\alpha(v,x)$.
    Otherwise, if $vx$ is not a border edge, insert an edge $vx$ to $G$
    and set $\alpha(v,x)=\alpha(u,x)$.
    If, on the other hand, $v\in B$ and $x\in B$ (i.e., $vx$ is a border edge), report
    $\alpha(u,x)$ as a border edge.
\end{itemize}
Observe that our order of updates issued to $G$ guarantees that $G$ remains planar at all times.

The decision whether we merge $u$ into $v$ or $v$ into $u$ heavily affects
both the correctness and efficiency of the data structure.
First, if one of $u,v$ (say $v$) is contained in~$B$, whereas the other (say $u$) is not,
we merge $u$ into $v$.
If, however, we have $\{u,v\}\subseteq B$ or $\{u,v\}\subseteq V\setminus B$, we 
pick a vertex (say $u$) with a smaller set $\cvert^{-1}(u)$ and merge $u$ into $v$.

To handle an insertion of a new edge $e=xy$, we first check whether
$xy$ is a border edge.
If so, we discard $e$ and report it.
Otherwise, check whether $x$ and $y$ are adjacent in $G$.
If so, report the parallelism $e\to\alpha(x,y)$.
If not, add an edge $xy$ to $G$ and set $\alpha(x,y)=e$.

\begin{lemma}\label{lem:brute_repr}
Let $G$ be a graph initially equal to a simple planar graph $G_0=(V_0,E_0)$ such that $n=|V_0|$.
There is a data structure for the bordered vertex merging problem that processes any sequence of modifications of $G_0$, along with reporting parallelisms and border edges, in $O((n+f)\log^2{n}+m)$ total time, where $m$
  is the total number of edge insertions and $f$ is the total
  number of insertions of edges connecting non-adjacent vertices.
\end{lemma}
\begin{proof}
Clearly, by Lemma~\ref{lem:adjq}, building the initial representation takes $O(n\log{n})$ time,
  as we insert $O(n)$ edges to $G$.
  The reporting of parallel edges and border edges
  takes $O(n+m)$ time,
  since each (initial or inserted)
  edge is reported as a border edge or occurs as the first coordinate
  of a reported directed parallelism at most once.
  
  Also note that, by Lemma~\ref{lem:adjq}, an insertion of a parallel edge costs $O(1)$ time,
  for a total of $O(m)$ time over all insertions, as $G$ is not updated
  in that case.
  Recall that, by Fact~\ref{fac:adj}, accessing and updating
  values $\alpha(x,y)$ for $xy\in E(G)$ takes $O(1)$ time.

  The total cost of maintaining the representation of $G$ is $O(g\log{n})$, where
  $g$ is the total number of edge updates to the data structure of Lemma~\ref{lem:adjq}.
  We prove that $g=O((n+f)\log{n})$.
  To this end, we look at the merge of $u$ into $v$ from a different perspective:
  instead of removing an edge $e=ux$ and inserting an edge $vx$,
  imagine that we simply change an endpoint $u$ of $e$
  to $v$, but the edge itself does not lose its identity.
  Then, new edges in $G$ are only created either during
  the initialization or by inserting an edge connecting
  the vertices that have not been previously adjacent in $G$.
  Hence, there are $O(n+f)$ creations of new edges.

  Consider some edge $e=xy$ of $G$ immediately after its creation.
  Denote by $q(e)$ the pair 
  \iffull
  \linebreak
  \fi
  $(|\cvert^{-1}(x)|,|\cvert^{-1}(y)|)$.
  The value of $q(e)$ always changes when some endpoint of $e$ is updated.
  Suppose a merge of $u$ into $v$ ($u\neq v$) causes the change of some
  endpoint $u$ of $e$ to $v$.
  We either we have $u\notin B$ and $v\in B$
  or $|\cvert^{-1}(v)|\geq |\cvert^{-1}(u)|$ before the merge.
  The former situation can arise at most once per each endpoint of $e$, since we always merge
  a non-border vertex into a border vertex,
  if such case arises.
  In the latter case, on the other hand, one coordinate of $q(e)$
  grows at least by a factor of $2$, and clearly this can happen at most $O(\log{n})$ times,
  as the size of any $\cvert^{-1}(x)$ is never more than $n$.
  Since there are $O(n+f)$ ``created'' edges, and each such edge undergoes
  $O(\log{n})$ endpoint updates, indeed we have $g=O((n+f)\log{n})$.

  A very similar argument can be used to show that the total time needed to maintain the mapping
  $\cvert$ along with the reverse mapping $\cvert^{-1}$ is $O(n\log{n})$.
\end{proof}
\ifshort
\vspace{-4mm}
\fi
\subparagraph{A Micro Data Structure.}
In order to obtain an optimal data structure, we need the following specialized version
of the bordered vertex merging data structure that handles very small
graphs in linear total time.
Suppose we disallow inserting new edges into $G$.
Additionally, assume we are allowed to perform some preprocessing in time $O(n)$.
Then, due to a monotonous nature of allowed operations on~$G$,
when the size of $G_0$ is very small compared to $n$,
we can maintain $G$ faster than by using the data structure
of Lemma~\ref{lem:brute_repr}.
\begin{lemma}\label{lem:micro}
  After preprocessing in $O(n)$ time, we can repeatedly solve
the bordered vertex merging problem without edge insertions
  for planar simple graphs $G_0$ with $t=O(\log^4{\log^4{n}})$ vertices
  in $O(t)$ time.
\end{lemma}
\iffull
\begin{proof}
Let $f(n)=c\log^4{\log^4{n}}$ for some $c>0$.
We use the preprocessing time to simulate every possible sequence
  of modifications on every possible graph $G_0=(V_0,E_0)$ with no more
  than $f(n)$
  vertices and each possible $B\subseteq V_0$.
  The simulation allows us to precompute for each step the list of self-loops and directed
  parallelisms to be reported.

  We identify the vertices $V_0$ with the set $\{1,\ldots,t\}$
  and assume that edges of $E_0$ are assigned identifiers
  from the set $1,\ldots,|E_0|$ such that $e=uv\in E_0$ is assigned
  an identifier equal to the position of the pair $(u,v)$
  in the sorted list $\{(u,v):u<v \land uv\in E_0\}$.

  Any possible graph $G_0$ can be encoded with $O(f(n)^2)$ bits,
  representing the adjacency matrix of $G_0$.
  For a given $G_0$ with $t$ vertices, each possible $B\subseteq V_0$ can
  be easily encoded with additional $O(t)=O(f(n))$ bits.
  On a graph $G$ initially equal to $G_0$, at most $t$ merges can be performed.
  Clearly, a single operation on $G$ can be encoded as a pair
  of affected vertices, i.e., $O(\log{t})$ bits.
  Each possible sequence $S$ of modifications of $G$ (not necessarily maximal) can be thus encoded
  with additional $O(t\log{t})=O(f(n)^2)$ bits.

  We conclude that each triple $(G_0,B,S)$ can be encoded with $O(\poly(f(n)))$ bits
  and thus there are no more than $O(2^{\poly(f(n))})$ such triples.

  For each triple $\psi=(G_0,B,S)$, we do the following:
  \begin{itemize}
    \item We compute its bit encoding $z(\psi)$.
  \item We use the data structure $D$ of Lemma~\ref{lem:brute_repr}
    to simulate the sequence of updates $S$ on a graph $G$ initially equal to
      $G_0$ and a border set $B$.
  \item Afterwards, a record $Q[z(\psi)]$ is filled with the following information:
      \begin{itemize}
        \item mappings $\cvert$ and $\cvert^{-1}$ computed by $D$,
        \item the lists of border edges and directed parallelisms that were reported
          after the last modification of the sequence $S$.
        \item the bit encodings $z(\psi')$ of all the triples $\phi'=(G_0,B,S')$,
          such that $S'$ extends $S$ by a single modification.
      \end{itemize}
  \end{itemize}
  For each triple $\psi=(G_0,B,S)$, all the needed information can be clearly
  computed in time polynomial in $f(n)$.
  Hence, in total we need $O(2^{\poly(f(n))})$ time to compute all the necessary
  information.
  As $O(\poly(f(n)))=o(\log{n})$
  any bit encoding $z(\psi)$ is an integer of order $O(n)$
  and fits in $O(1)$ machine words.

  Now, to handle any sequence of modifications on a graph $G_0$ with at most
  $f(n)$ vertices and a border set $B\subseteq V_0$, we first
  compute in linear time the bit encoding $z(\psi^{*})$ of $\psi^{*}=(G_0,B,S)$, where initially
  $S=\emptyset$.
  Each modification $Y$ is executed as follows: we use the information in $Q[z(\psi^*)]$
  to find the bit encoding $z(\psi')$ of the configuration $\psi'=(G_0,B,S\cup \{Y\})$
  and we move from the configuration $\psi^*$ to $\psi'$.
  Next, we read from $Q[z(\psi^{'})]$ which edges should be reported as parallel edges
  or border edges.
  As we only move between the configurations by updating the bit encoding of the current configuration
  and possibly report edges, the whole sequence of updates takes time linear in the size of $G_0$.
  Clearly, the record $Q[z(\psi^{*})]$ can be used to
  access the mappings $\cvert$ and $\cvert^{-1}$ in constant time.
\end{proof}
\fi

\subsection{A Multi-Level Data Structure}
Recall that our goal is to maintain $G$ under contractions.
Below we describe in detail how to take advantage of graph partitioning
and bordered vertex merging data structures to obtain a linear
time solution.
To simplify the further presentation, we assume that the initial version $G_0=(V_0,E_0)$
of $G$ is simple and of constant degree.
\ifshort
The standard reduction assuring that is described in the full version.
\fi
\iffull
The standard reduction assuring that is described in Appendix~\ref{a:degree}.
\fi

We build an $r$-division $\rdiv=\{P_1,P_2,\ldots,\}$ of $G$ with $r=\log^4{n}$, where $n = |V_0|$ (see Lemma~\ref{lem:rdiv}).
Then, for each piece $P_i\in\rdiv$, we build an $r$-division $\rdiv_i=\{P_{i,1},P_{i,2},\ldots\}$
of $P_i$ with $r=\log^4{\log^4{n}}$.
By Lemma~\ref{lem:rdiv}, building all the necessary pieces takes $O(n)$ time
in total.
Since $G_0$ is of constant degree, any vertex $v\in V_0$ is contained in $O(1)$ pieces of $\rdiv$.
Analogously, for any $v\in P_i$, $v$ is contained in $O(1)$ pieces of $\rdiv_i$.

\newcommand{\glob}{\ensuremath{\mathcal{G}}}
\newcommand{\loga}{\ensuremath{\mathcal{L}}}
\newcommand{\micr}{\ensuremath{\mathcal{M}}}
\newcommand{\gpart}{\ensuremath{\mathcal{F}}}
\newcommand{\dss}{\ensuremath{\Pi}}
\newcommand{\sstr}{\ensuremath{\mathcal{D}}}
\newcommand{\bds}{\ensuremath{\pi}}

As $G$ undergoes contractions, let $\cvert:V_0\to V$ be a mapping such that
for each $v\in V_0$, $v$~``has been merged'' into $\cvert(v)$.
As we later describe, a vertex resulting from contracting an edge $uv$ will be called either $u$ or $v$,
which guarantees that $V \subseteq V_0$ at all times.
Of course, initially $\cvert(v)=v$ for each $v\in V=V_0$.

\newcommand{\simp}[1]{\ensuremath{\overline{#1}}}
\newcommand{\parstr}{\ensuremath{par}}

Let $\simp{G}=(V,\simp{E})$ denote the maximal simple subgraph of $G$, i.e.,
the graph $G$ with self-loops discarded and each group $Y$ of parallel
edges replaced with a single edge $\alpha(Y)$.
The key component of our data structure is a 3-level set
of
(possibly micro-) bordered vertex merging data structures
$\dss=\{\bds\}\cup \{\bds_i:P_i\in\rdiv\}\cup \{\bds_{i,j}:P_i\in\rdiv, P_{i,j}\in\rdiv_i\}$.
The data structures $\dss$ form a tree such that $\bds$ is the root,
$\{\bds_i:P_i\in\rdiv\}$ are the children of $\bds$ and
$\{\bds_{i,j}:P_{i,j}\in\rdiv_i\}$ are the children
of $\bds_i$.
For $\sstr\in\dss$, let $\parstr(\sstr)$ be the parent of $\sstr$
and let $A(\sstr)$ be the set of ancestors of $\sstr$.
We call the value $h(\sstr)=|A(\sstr)|$ a \emph{level} of $\sstr$.
The data structures of levels $0$ and $1$ are stored as data structures
of Lemma~\ref{lem:brute_repr}, whereas the data structures of level
$2$ are stored as micro structures of Lemma~\ref{lem:micro}.

Each data structure $\sstr\in\dss$ has a defined set $V_\sstr\subseteq V_0$ of \emph{interesting vertices},
defined as follows: 
$V_\bds=\bnd{\rdiv}$, $V_{\bds_i}=\bnd{P_i}\cup \bnd{\rdiv_i}$ and $V_{\bds_{i,j}}=V(P_{i,j})$.
The data structure $\sstr$ maintains a certain subgraph $G_\sstr$ of $\simp{G}$ defined inductively as follows (recall that we define $G_1 \setminus G_2$ to be a graph containing all vertices of $G_1$ and edges of $G_1$ that do not belong to $G_2$)
$$G_\sstr=\simp{G}[\cvert(V_\sstr)]\setminus\Biggl(\bigcup_{\sstr'\in A(\sstr)}G_{\sstr'}\Biggr).$$
\begin{fact}\label{fac:minor}
For any $\sstr\in\dss$, $G_\sstr$ is a minor of $G_0$.
\end{fact}

\begin{fact}\label{fac:edge-exists}
  For any $uv=e\in \simp{E}$, there exists $\sstr\in\dss$ such that $e\in E(G_\sstr)$.
\end{fact}
\iffull
\begin{proof}
  Let $u_0,v_0$ be the initial endpoints of $e$. Initially $e\in P_{i,j}$ for some $i,j$. Observe that, since
  $\{\cvert(u_0),\cvert(v_0)\}\subseteq V(G_{\bds_{i,j}})$, $e$ is contained in
  $G_{\bds_{i,j}}$ or some of its ancestors.
\end{proof}
\fi

Each $\sstr\in\dss$ is initialized with the graph $G_\sstr$,
according to the initial mapping $\cvert(v)=v$ for any $v\in V_0$.
We define the set of \emph{ancestor vertices} $AV_\sstr=V_\sstr\cap\left(\bigcup_{\sstr'\in A(\sstr)} V_{\sstr'}\right)$.

Now we discuss what it means for the bordered vertex merging data structure $\sstr$ to
maintain the graph $G_\sstr$.
Note that the vertex set used to initialize $\sstr$ is $V_\sstr$.
We write $\cvert_\sstr, \cvert_\sstr^{-1}$ to denote
the mappings $\cvert, \cvert^{-1}$ maintained by $\sstr\in\dss$, respectively.
Throughout a sequence of contractions, we maintain the following invariants for any $\sstr\in\dss$:
\begin{itemize}
\item There is a 1-1 mapping between the sets $\cvert(V_\sstr)$ and $\cvert_\sstr(V_\sstr)$
  such that for the corresponding vertices $x\in \cvert(V_\sstr)$ and $y\in \cvert_\sstr(V_\sstr)$
    we have $\cvert^{-1}_\sstr(y)=\cvert^{-1}(x)\cap V_\sstr$.
    We also say that $x$ is \emph{represented} in $\sstr$ in this case.
  \item There is an edge $xy\in E(G_\sstr)$ if and only if there is an edge $x'y'$ in the graph
    maintained by $\sstr$, where $x',y'\in \cvert_\sstr(V_\sstr)$ are the corresponding
    vertices of $x$ and $y$, respectively.
\item The border set $B_\sstr$ of $\sstr$ is always equal to $\cvert_\sstr(AV_\sstr)$.
\end{itemize}
Thus, the graph maintained by $\sstr$ is isomorphic to $G_\sstr$ but can technically
use a different vertex set.
Observe that in $G_\sstr$ there are no edges between the vertices $\cvert(AV_\sstr)$ and 
the following fact describes how this is reflected in $\sstr$.
\begin{fact}\label{fac:bempty}
In the graph stored in $\sstr$, no two vertices of $B_\sstr$ are adjacent.
\end{fact}

Note that as the sets $V_\sstr$ and $V_{\sstr'}$ might overlap for $\sstr\neq\sstr'$,
the vertices of $V$ can be represented in multiple data structures.

\begin{lemma}\label{lem:uniq_vert}
  Suppose for $v\in V$ we have $v\in V(G_{\sstr_1})$ and $v\in V(G_{\sstr_2})$. Then, $v\in V(G_\sstr)$,
  where $\sstr$ is the lowest common ancestor of $\sstr_1$ and $\sstr_2$.
\end{lemma}
\iffull
\begin{proof}
  We first prove that for $i\neq j$, $\cvert(V(P_i))\cap \cvert(V(P_j))\subseteq\cvert(\bnd{\rdiv})$.
  Assume the contrary.
  Thus, there exists such $w\in \cvert(V(P_i))\cap \cvert(V(P_j))$ that $w\notin \cvert(\bnd{\rdiv})$.
  But for $x\in V$, $G_0[\cvert^{-1}(x)]$ is a connected subgraph of $G_0$
  and thus $G_0[\cvert^{-1}(w)]$ is connected and contains both some vertex of $P_i$
  and some vertex for $P_j$.
  But each path from $V(P_i)$ to $V(P_j)$ in $G_0$ has to go through a vertex of $\bnd{R}$,
  by the definition of an $r$-division.
  Hence $\bnd{R}\cap\cvert^{-1}(w)\neq\emptyset$ and $w\in\cvert(\bnd{\rdiv})$, a contradiction.

  Analogously one can prove that for any $i$ and $j\neq k$, $\cvert(V(P_{i,j}))\cap \cvert(V(P_{i,k}))\subseteq\cvert(\bnd{\rdiv_i})$.
  
  Suppose that $v\in V(G_{\sstr_1})\cap V(G_{\sstr_2})$. If for some $i\neq j$ we have
  $\sstr_1\in\{\bds_i\}\cup\bigcup_k\{\bds_{i,k}\}$ and $\sstr_2\in\{\bds_j\}\cup\bigcup_k\{\bds_{j,k}\}$
  then $v\in\cvert(P_i)\cap\cvert(P_j)$ and we conclude $v\in\cvert(\bnd{R})$ and hence $v\in V(G_\bds)$.
  Analogously we prove that if $\sstr_1=\bds_{i,j}$ and $\sstr_2=\bds_{i,k}$ for some $j\neq k$, then
  $v\in V(G_{\bds_i})$.
\end{proof}
\fi

By Lemma~\ref{lem:uniq_vert}, each vertex $v\in V$ is represented in a unique
data structure of minimal level, a lowest common ancestor of all data structures
where $v$ is represented.
We denote such a data structure by $\sstr(v)$.
Observe that for any $\sstr\in\dss$ the vertices $\{v:\sstr(v)=\sstr\}$ are represented
in $\sstr$ by $\cvert_{\sstr}(V_\sstr)\setminus \cvert_{\sstr}(AV_\sstr)$.

We now describe the way we index the vertices of $V$. This is required, as upon a contraction,
our data structure returns an identifier of a new vertex.
We also reuse the names of the initial vertices $V_0$, as the
bordered vertex merging data structures do.
Namely, a vertex $v\in V$ is labeled with $\cvert_{\sstr(v)}(v')\in V_0$, where $v'$ represents
$v$ in $\sstr(v)$.

Note that, as the bordered vertex merging data structures always merge one vertex
involved into the other, for any $\sstr\in\dss$ we have
$\cvert_\sstr(V_\sstr)\setminus\cvert_\sstr(AV_\sstr)\subseteq V_\sstr\setminus AV_\sstr$.
Hence the label sets used by distinct sets $\{v:\sstr(v)=\sstr\}$ are distinct,
since the sets of the form $V_\sstr\setminus AV_\sstr$ are pairwise disjoint.
Such a labeling scheme makes it easy to find the data structure
$\sstr(v)$ by looking only at the label.

For brevity, in the following we sometimes do not distinguish between 
the set $V$ and the set of labels $\bigcup_{\sstr\in\dss}\left(\cvert_\sstr(V_\sstr)\setminus\cvert_\sstr(AV_\sstr)\right)$.

\begin{lemma}\label{lem:gedge}
  Let $uv=e\in\simp{E}$ and $h(\sstr(u))\geq h(\sstr(v))$.
  Then $e\in E(G_{\sstr(u)})$ and either $\sstr(u)=\sstr(v)$ or $\sstr(u)$ is a descendant of $\sstr(v)$.
\end{lemma}
\iffull
\begin{proof}
  If $\{u,v\}\subseteq\cvert(\bnd{R})$, then clearly $h(\sstr(u))=h(\sstr(v))=0$, $e\in E(G_\bds)$
  and the lemma holds.
  Moreover, no $G_\sstr$ such that $\sstr$ is a descendant of $\bds$ can contain the edge $uv$.
  
  Otherwise, $u$ does not belong to $\cvert(\bnd{R})$.
  Consequently, by Fact~\ref{fac:edge-exists}~and~Lemma~\ref{lem:uniq_vert}, there exists exactly one $i$ such that $\cvert(u)$
  is a vertex of some graph $G_\sstr$, and $e$ is an edge of some graph $G_{\sstr'}$, 
  where $\sstr,\sstr'$ are data structures in the subtree of $\dss$ rooted of $\bds_i$.
  If $v\notin \cvert(\bnd{R})$, $v$ cannot be a vertex of any $G_{\sstr''}$, where $\sstr''$
  is in the subtree of $\bds_j$, $j\neq i$.

  Again, if $\{u,v\}\subseteq\cvert(\bnd{P_i}\cup\bnd{\rdiv_i})$, then
  $uv\in E(G_{\bds_i})$, $\sstr(u)=\bds_i$ and no descendant of $G_{\bds_i}$ can contain $uv$.
  If not, we analogously get that there might exist at most one $G_{\bds_{i,j}}$
  containing the edge $uv$ and the vertex $u$.
  If $v\notin \cvert(\bnd{P_i}\cup\bnd{\rdiv_i})$, then only $G_{\bds_{i,j}}$ can
  contain the vertex $v$.

  In all cases, $\sstr(u)=\sstr(v)$ or $\sstr(u)$ is a descendant of $\sstr(v)$.
\end{proof}
\fi

\begin{lemma}\label{lem:bind}
Let $uv$ be an edge of some $G_\sstr$, $\sstr\in\dss$.
  If $\{u,v\}\subseteq V(G_{\sstr'})$, where $\sstr'\neq\sstr$, then
  $\sstr'$ is a descendant of $\sstr$ and
  both $u$ and $v$ are represented as border vertices of $\sstr'$.
\end{lemma}
\iffull
\begin{proof}
  Suppose wlog. that $h(\sstr(u))\geq h(\sstr(v))$.
  By Lemma~\ref{lem:gedge}, we have $\sstr=\sstr(u)$.

  Let  $\{u,v\}\subseteq V(G_{\sstr'})$.
If $\sstr'$ is a descendant of $\sstr$,
  then $\{u,v\}\subseteq \cvert(AV_{\sstr'})$ and by the invariants
  maintained by our data structure, $u$ and $v$ are represented
  by the vertices of $B_{\sstr'}$.

Suppose $\sstr'\neq \sstr$ and $\sstr'$ is not a descendant of $\sstr$.
  Then, by Lemma~\ref{lem:uniq_vert}, the lowest common ancestor of $\sstr'$ and $\sstr$ contains
  the vertex $u$ and is an ancestor of $\sstr=\sstr(u)$, a contradiction.
\end{proof}
\fi

\begin{lemma}\label{lem:children}
  Let $v\in \cvert(V_\sstr)$, where $\sstr\in\dss$. Then, $v$ is represented in $O(|\cvert^{-1}_\sstr(v)|)$
  data structures $\sstr'$ such that $\parstr(\sstr')=\sstr$.
\end{lemma}
\iffull
\begin{proof}
Let $\sstr'$ be a child of $\sstr$. If $v$ is represented in $\sstr'$, then
  $v\in \cvert(V_\sstr)\cap\cvert(V_{\sstr'})$.
  It follows that as $G_0[\cvert^{-1}(v)]$ is a connected subgraph of $G_0$,
  it contains a path between some vertex $x\in V_\sstr$ and some vertex $y\in V_{\sstr'}$.
  Assume $x\notin V_{\sstr'}$.
  If $\sstr'=\bds_i$, then in fact we have $x\in V(P_j)$, for $j\neq i$ and any path
  from $x$ to $y$ has to go through a vertex of $z\in \bnd{P_i}$ and as $\bnd{P_i}\subseteq V_\sstr\cap V_{\sstr'}$,
  there exists a vertex from $V_\sstr\cap V_{\sstr'}$ in $\cvert^{-1}(v)$.
  Similarly, if $\sstr'=\bds_{i,j}$, there exists a vertex of $\bnd{P_{i,j}}$
  in $\cvert^{-1}(v)$ and we again obtain $\cvert^{-1}(v)\cap V_\sstr\cap V_{\sstr'}\neq\emptyset$.

  Recall that we maintain an invariant $\cvert^{-1}_\sstr(v)=\cvert^{-1}(v)\cap V_\sstr$.
  Hence, $\cvert^{-1}_\sstr(v)\cap V_{\sstr'}\neq\emptyset$.
  However, for each $w\in \cvert^{-1}_\sstr(v)$, there are only $O(1)$ child data structures $\sstr'$
  such that $w\in V_\sstr'$, by the constant degree assumption.
  It follows that there are can be at most $O(|\cvert^{-1}_\sstr(v)|)$ data structures $\sstr'$
  such that $\parstr(\sstr')=\sstr$ and $\cvert^{-1}_\sstr(v)\cap V_{\sstr'}\neq\emptyset$,
  which in turn means that there are at most $O(|\cvert^{-1}_\sstr(v)|)$ data structures $\sstr'$
  such that $v$ is also represented in $\sstr'$.
\end{proof}
\fi
We also use the following auxiliary components for each $\sstr\in\dss$:

\newcommand{\parptr}{\ensuremath{\beta}}
\newcommand{\chptr}{\ensuremath{\gamma}}

\begin{itemize}
  \item For each $x\in \cvert_{\sstr}(AV_\sstr)$ we maintain a pointer
    $\parptr_\sstr(x)$ into $y\in \cvert_{\parstr(\sstr)}(AV_\sstr)$,
    such that $x$ and $y$ represent
    the same vertex of the maintained graph $G$.
  \item A dictionary (we use a balanced BST) $\chptr_\sstr$ mapping a pair $(\sstr',x)$, where $\sstr'$ is a child of $\sstr$
    and $x\in \cvert_\sstr(V_\sstr)$, to a vertex $y\in \cvert_{\sstr'}(AV_\sstr)$ iff $x$ and $y$
    represent the same vertex of $V$.
\end{itemize}

Another component of our data structure
is the forest $\parf$ of reported parallelisms:
for each reported parallelism $e\to \alpha(e)$, we make $e$ a child of $\alpha(e)$
in $\parf$.
Note that the forest $\parf$ allows us to go through all the edges parallel to $\alpha(e)$
in time linear in their number.

\begin{lemma}\label{lem:access-mapping}
  For $v_0\in V_0$, we can compute $\cvert(v_0)$ and find $\sstr(\cvert(v_0))$ in $O(1)$ time.
\end{lemma}
\iffull
\begin{proof}
Let $P_{i,j}$ be any piece such that $v_0\in V(P_{i,j})$.
  First, we can compute the representation $x=\cvert_{\bds_{i,j}}(v_0)$ 
  of $\cvert(v_0)$ in $\bds_{i,j}$
  in $O(1)$ time, as the data structure $\bds_{i,j}$ stores
  the mapping $\cvert_{\bds_{i,j}}$ explicitly.
  Set $\sstr=\bds_{i,j}$.
  
  Next, if $x\in\cvert_{\sstr}(AV_\sstr)$ (or, technically speaking,
  if $x\in AV_\sstr$), we follow the pointer $\parptr_\sstr(x)$ to
  the data structure of lower level and repeat if needed,
  until we reach the data structure $\sstr(\cvert(v_0))$.
  As the tree of data structures has $3$ levels, we follow
  $O(1)$ pointers.
\end{proof}
\fi
\begin{lemma}\label{lem:transl}
  Let $v\in V(G_\sstr)$. For any $\sstr'$, such that $\parstr(\sstr')=\sstr$,
  we can compute the vertex $v'$
  representing $v$ in $G_{\sstr'}$ (or detect that such $v'$ does not exist)
  in $O(\log|V_\sstr|)$ time.
\end{lemma}
\iffull
\begin{proof}
  By Lemma~\ref{lem:children}, the number of entries in $\chptr_\sstr$ is
  $O\left(\sum_{v\in V(G_\sstr)}|\cvert^{-1}_\sstr(v)|\right)=O(V_\sstr)$.
  The cost of any operation on a balanced binary search is logarithmic in the size
  of the tree.
\end{proof}
\fi

We now describe how to implement the call $(s,P,L):=\dscontr(e)$, where $uv=e\in E$, $u,v\in V$.
Suppose the 
initial endpoints of $e$ were $u_0,v_0\in V_0$.
First, we iterate through the tree $T_e\in \parf$ containing $e$
to find $\alpha(e)$.
By Lemma~\ref{lem:access-mapping}, we can find the vertices $u,v$ along with
the respective data structures $\sstr(u),\sstr(v)$, based on $u_0,v_0$
in $O(1)$ time.
Assume wlog. that $h(\sstr(u))\geq h(\sstr(v))$.
By Lemma~\ref{lem:gedge}, $\alpha(e)$ is an edge of $G_{\sstr(u)}$.
Although we are asked to contract $e$, we conceptually contract $\alpha(e)$, by
issuing a merge of $u$ and $v$ to $\sstr(u)$.
To reflect that we were actually asked to contract $e$, we include
all the edges of $T_e\setminus\{e\}$ in $L$ as self-loops.
The merge might make $\sstr(u)$ report some parallelisms $e_1\to e_2$.
In such a case we report $e_1\to e_2$ to the user (by including it in $P$) and update the forest $\parf$.

We now have to reflect the contraction of $e$ in all the
required data structures $\sstr\in\dss$, so that our invariants
are satisfied.
Assume wlog. that $u$ is merged into $v$ in $\sstr$.
If before the contraction, both $u$ and $v$ were the vertices
of some $G_{\sstr'}$, $\sstr'\neq\sstr$, then
by Lemma~\ref{lem:bind}, $\sstr'$ is a descendant of $\sstr$. 
By a similar argument as in the proof of Lemma~\ref{lem:brute_repr},
we can afford to iterate through $\cvert_{\sstr}^{-1}(u)$ without increasing the asymptotic
performance of the $u$-into-$v$ merge performed by $\sstr$,
as long as we spend $O(\log|V_\sstr|)$ time per element of $\cvert_{\sstr}^{-1}(u)$.
By Lemma~\ref{lem:children},
there are $O(|\cvert_{\sstr}^{-1}(u)|)$ data structures
$\sstr_1,\sstr_2,\ldots$ that are the children of $\sstr$
and contain the representation of $u$.
For each such $\sstr_i$, we first use the dictionary $\chptr_\sstr$
to find the vertex $x$ representing $u$ in $\sstr_i$,
%We then 
and update $\parptr_{\sstr_i}(x)$ to $v$.
Then, using Lemma~\ref{lem:transl}, we check whether $v\in V(G_{\sstr_i})$ in $O(\log|V_\sstr|)$ time.
If not, we set $\chptr_\sstr(\sstr_i,v)$ to $x$.
Otherwise, we merge $u$ and $v$ in $\sstr_i$ and handle this
merge -- in terms of updating the auxiliary components $\parptr$ and $\chptr$
-- analogously as for~$\sstr$.
This is legal, as $u,v\in\cvert_{\sstr_i}(AV_{\sstr_i})$ and
thus $u$ and $v$ are border vertices in $\sstr_i$, by Fact~\ref{fac:bempty}.
The merge may cause $\sstr_i$ to report some parallelisms.
We handle them as described above in the case of the data structure~$\sstr$.
Note however 
that merging border vertices 
cannot cause reporting of new border edges (i.e., those with both
endpoints in $B_{\sstr_i}$).

The merge of $u$ and $v$ in $\sstr$ might also create some new edges $e'=xy$ between
the vertices $\cvert_\sstr(AV_\sstr)$ in $G_\sstr$.
Note that in this case $\sstr$ reports $xy$ as a border edge
and also we know that $h(\sstr(x))<h(\sstr)$ and $h(\sstr(y))<h(\sstr)$.
Hence, $e'$ should end up in some of the ancestors of $\sstr$.
We insert $e'$ to $\parstr(\sstr)$.
$\parstr(\sstr)$ might also report $xy$ as a border edge and in
that case $e'$ is inserted to the grandparent of $\sstr$.
It is also possible that $e'$ will be reported a parallel edge
in some of the ancestors of $\sstr$: in such a case an appropriate
directed parallelism is added to $P$.

Note that all the performed merges and edge insertions
are only used to make the graphs represented by the data structures
satisfy their definitions.
  Fact~\ref{fac:minor} implies that the represented graphs
  remain planar at all times.

We now describe how the other operations are implemented.
To compute $u,v\in V$ such that $\{u,v\}=\dsverts(e)$, where $e\in E$,
we first use Lemma~\ref{lem:access-mapping} to compute
$u=\cvert(u_0)$ and $v=\cvert(v_0)$, where $u_0,v_0$ are the initial endpoints
of $e$.
Clearly, this takes $O(1)$ time.

To maintain the values $\dsdeg(v)$ of each $v\in V$, we simply set 
$\dsdeg(s):=\dsdeg(u)+\dsdeg(v)-1$ after a call $(s,P,L):=\dscontr(e)$.
Additionally, for each directed parallelism $e_1\to e_2$ we decrease
$\dsdeg(x)$ and $\dsdeg(y)$ by one, where $\{x,y\}=\dsverts(e_1)$.

\newcommand{\elist}{\ensuremath{\mathcal{E}}}

For each $u\in V$ we maintain a doubly-linked list $\elist(u)=\{\alpha(uv):uv\in \simp{E}\}$.
Additionally, for each $e\in \simp{E}$ we store the pointers to the two occurrences
of $e$ in the lists $\elist(\cdot)$.
Again after a call $(s,P,L):=\dscontr(e)$, where $e=uv$, we set $\elist(s)$ to be a concatenation
of the lists $\elist(u)$ and $\elist(v)$.
Finally, we remove all the occurrences of edges $\{\alpha(e)\}\cup \{e_1:(e_1\to e_2)\in P\}$
from the lists $\elist(\cdot)$.
Now, the implementation of the iterator $\dsnei(u)$ is easy,
as the endpoints not equal to $u$ of the edges in $\elist(u)$ form exactly the set $N(u)$.

\begin{lemma}\label{lem:easy-ops}
  The operations $\dsverts$, $\dsdeg$ and $\dsnei$ run in $O(1)$ worst-case time.
\end{lemma}

To support the operation $\dsedge(u,v)$ in $O(1)$ time, we first turn all the
dictionaries $\chptr_\sstr$ into hash tables with $O(1)$
expected update time and $O(1)$ worst-case query time \cite{Dietzfelbinger:1994}.
Our data structure thus ceases to be deterministic, but
we obtain a more efficient version of Lemma~\ref{lem:transl}
that allows us to compute the representation of a vertex
in a child data structure $\sstr'$ in $O(1)$ time.
By Lemma~\ref{lem:gedge}, the edge $uv$ can be contained in
either $\sstr(u)$ or $\sstr(v)$, whichever has greater level.
Wlog. suppose $h(\sstr(u))\geq h(\sstr(v))$.
Again, by Lemma~\ref{lem:gedge}, $\sstr(u)$ is a descendant
of $\sstr(v)$.
Thus, we can find $v$ in $\sstr(u)$ by applying
Lemma~\ref{lem:transl} at most twice.

\begin{lemma}\label{lem:edge-op}
If the dictionaries $\chptr_\sstr$ are implemented as hash tables,
  the operation $\dsedge$ runs in $O(1)$ worst-case time.
\end{lemma}
\iffull
The following lemma summarizes the total time spent
on updating all the vertex merging data structures $\dss$ and is proved
in Section~\ref{sec:time_proof}.
\fi
\begin{lemma}\label{lem:time}
  The cost of all operations on the data structures $\sstr\in\dss$ is $O(n)$.
\end{lemma}
\begin{proof}[Proof of Theorem~\ref{thm:main}]
To initialize our data structure, we initialize all the data structures $\sstr\in\dss$
  and the auxiliary components. This takes $O(n)$ time.
The time needed to perform any sequence of operations $\dscontr$ is proportional
to the total time used by the data structures $\dss$, as the cost of maintaining
the auxiliary components can be charged to the operations performed by
the individual structures of $\dss$.
By Lemma~\ref{lem:time}, this time is $O(n)$.
If the dictionaries $\chptr_\sstr$ are implemented as hash tables, this
bound is valid only in expectation.

By combining the above with Lemmas~\ref{lem:easy-ops}~and~\ref{lem:edge-op}, the theorem follows.
\end{proof}

\iffull
\subsection{Running Time Analysis}\label{sec:time_proof}
To bound the operating time of our data structure, we need to analyze,
for any $\sstr\in\dss$ and any sequence
of edge contractions, the number of changes to $E(G_\sstr)$
that result in a costly operation of inserting an edge connecting
non-adjacent vertices into the underlying bordered vertex merging
data structure $\sstr$.

\begin{lemma}\label{lem:total_ins}
  Let $\sstr\in\dss$.
  After the initialization of $\sstr$, only
  the edges initially contained in the graphs represented by the descendants of $\sstr$
  might be inserted into $\sstr$, each at most once.
\end{lemma}
\begin{proof}
  Note that whenever we report a border edge in $\sstr$, we insert it to $\parstr(\sstr)$.
\end{proof}

Consider some sequence $S$ of $k$ edge contractions on $G$.
Let $G_i=(V_i,E_i)$ (for $i=0,1,\ldots,k$) be the graph $G$ after $i$
contractions.
Denote by $u_i,v_i\in V_{i-1}$ the vertices involved in the $i$-th
contraction, and by $s_i\in V_i$ the vertex of $G_i$ obtained as a result
of the $i$-th contraction.
We have $V_i=V_{i-1}\setminus\{u_i,v_i\}\cup\{s_i\}$.
Moreover, let $\cvert_i:V_0\to V_i$ be the mapping $\cvert$ after $i$ contractions of $S$.
Denote by $\simp{G_i}$ the graph $\simp{G}$ after $i$ contractions.

Let $W\subseteq V_0$. For $i>0$, we define the set $\Delta_i^W\subseteq E(\simp{G_i})$ of ``new'' edges
appearing in the induced subgraph $\simp{G}[\cvert(W)]$ as a result
of the $i$-th contraction, in the following sense.
An edge $s_iy_i\in E(\simp{G_i})$ is included in $\Delta_i^W$
iff $\{s_i,y_i\}\subseteq \cvert_i(W)$ and:
\begin{itemize}
  \item either $u_i\notin \cvert_{i-1}(W)$ or $u_iy_i\notin E(\simp{G_{i-1}})$,
  \item either $v_i\notin \cvert_{i-1}(W)$ or $v_iy_i\notin E(\simp{G_{i-1}})$.
\end{itemize}
Note that this definition implies $y_i\in \cvert_{i-1}(W)$ and $|\{u_i,v_i\}\cap \cvert_{i-1}(W)|=1$.
Define
$\Psi_W=|\Delta_1^W|+|\Delta_2^W|+\ldots+|\Delta_k^W|.$
\begin{corollary}\label{cor:psi_meaning}
We have $\Psi_W=\sum_{i=1}^k d^W_i$, where $d^W_i$
  is the number of edges that should be added to $\simp{G_{i-1}}[\cvert_{i-1}(W)]$,
  after possibly performing a contraction of an edge $u_iv_i$ in $\simp{G_{i-1}}$
  (if $\{u_i,v_i\}\cap\cvert_{i-1}(W)\neq\emptyset$), in order to obtain $\simp{G_i}[\cvert_i(W)]$.
\end{corollary}

\begin{lemma}\label{lem:subg_bound}
  For any $W\subseteq V_0$, $\Psi_W=O(|W|)$.
\end{lemma}
\begin{proof}
  Fix some plane embedding of $G_0$. We define semi-strict versions $G_0^W,G_1^W,\ldots,G_k^W$
  of graphs $G_0[\cvert(W)],G_1[\cvert(W)],\ldots,G_k[\cvert(W)]$, respectively, so that:
  \begin{itemize}
    \item $G_0^W=G_0[\cvert(W)]$. Recall that $G_0$ is simple, and thus its subgraph $G_0[\cvert(W)]$
      is also simple and in particular semi-strict.
    \item If $\{u_i,v_i\}\cap \cvert_{i-1}(W)=\emptyset$, then $G_i^W=G_{i-1}^W$.
    \item If $\{u_i,v_i\}\subseteq \cvert_{i-1}(W)$, then we obtain $G_i^W$ from $G_{i-1}^W$
      by first contracting $u_iv_i$.
      For any triangular face $f=u_iv_ix_i$ of $G_{i-1}^W$ (there can be between $0$ and $2$ such faces),
    the contraction introduces a face $f'=s_ix_i$ of length $2$.
      We remove one of these edges $s_ix_i$ from $G_i^W$ so that the face $f'$ is merged with
      any neighboring face
      and $G_i^W$ is semi-strict.
    \item If $|\{u_i,v_i\}\cap \cvert_{i-1}(W)|=1$, suppose wlog. that $u_i\in \cvert_{i-1}(W)$
      (the case $v_i\in\cvert_{i-1}(W)$ is symmetrical).
      Pick a maximal pairwise non-parallel subset $F_i$ of such edges $v_ib_i$ of $G_{i-1}$,
      that $b_i\in\cvert_{i-1}(W)$ and $u_ib_i\notin E(G_{i-1}^W)$.
      Let $G_i^W$ be obtained from the following subgraph of $G_{i-1}$:
      $$X_i=(\cvert_{i-1}(W)\cup\{v_i\},E(G_{i-1}^W)\cup\{u_iv_i\}\cup F_i)$$
      by contraction of $u_iv_i$ (which merges vertices $u_i,v_i$ into $s_i$).
      Observe that, by definition of $X_i$ the contraction of $u_iv_i$ in $X_i$ does not
      introduce parallel edges and as a result $G_i^W$ is semi-strict.
      \end{itemize}
  The graphs $G_i^W$ are defined in such a way that for any $x,y\in \cvert_i(W)$,
  $xy\in E(\simp{G_i})$ if and only if $xy\in E(G_i^W)$.
  As a result, we have $\Delta^W_i\subseteq E(G^W_i)\setminus E(G^W_{i-1})$.
  It is thus sufficient to prove
  $$\Psi_W\leq |E(G^W_1)\setminus E(G^W_{0})|+|E(G^W_2)\setminus E(G^W_{1})|+\ldots+|E(G^W_k)\setminus E(G^W_{k-1})|=O(|W|).$$
  As each $G^W_{i}$ is semi-strict, $|E(G^W_i)|\leq 3|V(G^W_i)|=3|\cvert_i(W)|\leq 3|W|$.
  Moreover, as any contraction in a semi-strict graph decreases the number
  of edges by at most $3$, $|E(G^W_{i-1})\setminus E(G^W_i)|\leq 3$.
  In fact, by the definition of $G^W_i$, we have
  $E(G^W_{i-1})\not\subseteq E(G^W_i)$ if and only if $\{u_i,v_i\}\subseteq\cvert_{i-1}(W)$,
  i.e., when $|V(G^W_i)|<|V(G^W_{i-1})|$.
  This may happen for at most $|W|$ values of $i$, as $|V(G^W_0)|=W$.
  Denote the set of these values $i$ as $I$.
  
  We have
\begin{align*}
  \Psi_W &\leq\sum_{i=1}^k|E(G^W_i)\setminus E(G^W_{i-1})|\\
         &=\sum_{i=1}^k|E(G^W_i)\setminus (E(G^W_i)\cap E(G^W_{i-1}))|
         =\sum_{i=1}^k|E(G^W_i)|-|E(G^W_i)\cap E(G^W_{i-1})|\\
         &\leq \sum_{i=1}^k|E(G^W_i)|-\sum_{i\in\{1,\ldots,k\}\setminus I}|E(G^W_{i-1})|-\sum_{i\in I}(|E(G^W_{i-1})|-3)\\
         &=|E(G^W_k)|-|E(G^W_0)|+3|I|
         \leq 6|W|=O(|W|).\qedhere
\end{align*}
\end{proof}

\begin{proof}[Proof of Lemma~\ref{lem:time}]
Recall that by Lemma~\ref{lem:brute_repr}, the cost of any sequence of operations
  on $\sstr\in \{\bds\}\cup\{\bds_i:P_i\in\rdiv\}$ is
  $O((|V_\sstr|+f_\sstr)\log^2{|V_\sstr|}+m_\sstr)$, where
$m_\sstr$ is the total number of times an edge is inserted into $\sstr$
and $f_\sstr$ is the number of insertions connecting non-adjacent vertices.
  By Lemma~\ref{lem:total_ins}, $m_\bds=O(|E_0|)$ and $m_{\bds_i}=O(|E(P_i)|)$.
  By Corollary~\ref{cor:psi_meaning} and Lemma~\ref{lem:subg_bound},
  $f_\sstr=\Psi_{V_\sstr}=O(|V_\sstr|)$.
  We have $|V_\bds|=O(n/\log^2{n})$  and thus the cost of operating $\bds$ is $O(n)$.
  Similarly, we have $|V_{\bds_i}|=O(\log^4{n}/\log^2{\log^4{n}})$
  and the total cost of operating $O(n/\log^4{n})$ data structures $\bds_i$ 
  is $O(n/\log^2{\log^4{n}}+\sum_i|E(P_i)|)=O(n)$.

  By Lemma~\ref{lem:micro}, after $O(n)$ preprocessing, the total cost
  of operating each $\bds_{i,j}$ is $O(|V(P_{i,j})|)$ and thus, summed
  over all $i,j$, we again obtain $O(n)$ time.
\end{proof}

\fi

\bibliographystyle{plainurl}
\bibliography{references}

\newcommand{\sortkey}[1]{}
\begin{thebibliography}{10}

\bibitem{Brodal:1999}
Gerth~St{\o}lting Brodal and Rolf Fagerberg.
\newblock Dynamic representation of sparse graphs.
\newblock In {\em Algorithms and Data Structures, 6th International Workshop,
  {WADS} '99, Vancouver, British Columbia, Canada, August 11-14, 1999,
  Proceedings}, pages 342--351, 1999.
\newblock \href {http://dx.doi.org/10.1007/3-540-48447-7_34}
  {\path{doi:10.1007/3-540-48447-7_34}}.

\bibitem{Chechik:2017}
Shiri Chechik, Thomas~Dueholm Hansen, Giuseppe~F. Italiano, Veronika
  Loitzenbauer, and Nikos Parotsidis.
\newblock Faster algorithms for computing maximal 2-connected subgraphs in
  sparse directed graphs.
\newblock In {\em Proceedings of the Twenty-Eighth Annual {ACM-SIAM} Symposium
  on Discrete Algorithms, {SODA} 2017, Barcelona, Spain, Hotel Porta Fira,
  January 16-19}, pages 1900--1918, 2017.
\newblock \href {http://dx.doi.org/10.1137/1.9781611974782.124}
  {\path{doi:10.1137/1.9781611974782.124}}.

\bibitem{Dietzfelbinger:1994}
Martin Dietzfelbinger, Anna~R. Karlin, Kurt Mehlhorn, Friedhelm {Meyer auf der
  Heide}, Hans Rohnert, and Robert~Endre Tarjan.
\newblock Dynamic perfect hashing: Upper and lower bounds.
\newblock {\em {SIAM} J. Comput.}, 23(4):738--761, 1994.
\newblock \href {http://dx.doi.org/10.1137/S0097539791194094}
  {\path{doi:10.1137/S0097539791194094}}.

\bibitem{Edmonds65paths}
Jack Edmonds.
\newblock Paths, trees and flowers.
\newblock {\em Canadian Journal of Mathematics}, pages 449--467, 1965.

\bibitem{Frederickson84}
Greg~N. Frederickson.
\newblock On linear-time algorithms for five-coloring planar graphs.
\newblock {\em Inf. Process. Lett.}, 19(5):219--224, 1984.
\newblock \href {http://dx.doi.org/10.1016/0020-0190(84)90056-5}
  {\path{doi:10.1016/0020-0190(84)90056-5}}.

\bibitem{Gabow:2001}
Harold~N. Gabow, Haim Kaplan, and Robert~Endre Tarjan.
\newblock Unique maximum matching algorithms.
\newblock {\em J. Algorithms}, 40(2):159--183, 2001.
\newblock Announced at STOC '99.
\newblock \href {http://dx.doi.org/10.1006/jagm.2001.1167}
  {\path{doi:10.1006/jagm.2001.1167}}.

\bibitem{Giammarresi:96}
Dora Giammarresi and Giuseppe~F. Italiano.
\newblock Decremental 2- and 3-connectivity on planar graphs.
\newblock {\em Algorithmica}, 16(3):263--287, 1996.
\newblock \href {http://dx.doi.org/10.1007/BF01955676}
  {\path{doi:10.1007/BF01955676}}.

\bibitem{Goodrich:95}
Michael~T. Goodrich.
\newblock Planar separators and parallel polygon triangulation.
\newblock {\em J. Comput. Syst. Sci.}, 51(3):374--389, 1995.
\newblock \href {http://dx.doi.org/10.1006/jcss.1995.1076}
  {\path{doi:10.1006/jcss.1995.1076}}.

\bibitem{Gustedt}
Jens Gustedt.
\newblock Efficient union-find for planar graphs and other sparse graph
  classes.
\newblock {\em Theor. Comput. Sci.}, 203(1):123--141, 1998.
\newblock \href {http://dx.doi.org/10.1016/S0304-3975(97)00291-0}
  {\path{doi:10.1016/S0304-3975(97)00291-0}}.

\bibitem{Hopcroft:74}
John~E. Hopcroft and Robert~Endre Tarjan.
\newblock Efficient planarity testing.
\newblock {\em J. {ACM}}, 21(4):549--568, 1974.
\newblock \href {http://dx.doi.org/10.1145/321850.321852}
  {\path{doi:10.1145/321850.321852}}.

\bibitem{Karger:1993}
David~R. Karger.
\newblock Global min-cuts in rnc, and other ramifications of a simple min-out
  algorithm.
\newblock In {\em Proceedings of the Fourth Annual ACM-SIAM Symposium on
  Discrete Algorithms}, SODA '93, pages 21--30, Philadelphia, PA, USA, 1993.
  Society for Industrial and Applied Mathematics.

\bibitem{Klein:book}
Philip~N. Klein and Shay Mozes.
\newblock Optimization algorithms for planar graphs, 2017.
\newblock URL: \url{http://planarity.org}.

\bibitem{Klein:13}
Philip~N. Klein, Shay Mozes, and Christian Sommer.
\newblock Structured recursive separator decompositions for planar graphs in
  linear time.
\newblock In {\em Symposium on Theory of Computing Conference, STOC'13, Palo
  Alto, CA, USA, June 1-4, 2013}, pages 505--514, 2013.
\newblock \href {http://dx.doi.org/10.1145/2488608.2488672}
  {\path{doi:10.1145/2488608.2488672}}.

\bibitem{Lacki:2015}
{\sortkey{Lzzze}}{Jakub Łącki and Piotr Sankowski}.
\newblock Optimal decremental connectivity in planar graphs.
\newblock In {\em 32nd International Symposium on Theoretical Aspects of
  Computer Science, {STACS} 2015, March 4-7, 2015, Garching, Germany}, pages
  608--621, 2015.
\newblock \href {http://dx.doi.org/10.4230/LIPIcs.STACS.2015.608}
  {\path{doi:10.4230/LIPIcs.STACS.2015.608}}.

\bibitem{Mares}
Martin Mare{\v{s}}.
\newblock Two linear time algorithms for mst on minor closed graph classes.
\newblock {\em Archivum mathematicum}, 40(3):315--320, 2002.

\bibitem{Matsui}
Tomomi Matsui.
\newblock The minimum spanning tree problem on a planar graph.
\newblock {\em Discrete Applied Mathematics}, 58(1):91--94, 1995.
\newblock \href {http://dx.doi.org/10.1016/0166-218X(94)00095-U}
  {\path{doi:10.1016/0166-218X(94)00095-U}}.

\bibitem{Tarjan80}
David~W. Matula, Yossi Shiloach, and Robert~E. Tarjan.
\newblock Two linear-time algorithms for five-coloring a planar graph.
\newblock Technical report, Stanford University, Stanford, CA, USA, 1980.

\bibitem{Robertson:1996}
Neil Robertson, Daniel~P. Sanders, Paul Seymour, and Robin Thomas.
\newblock Efficiently four-coloring planar graphs.
\newblock In {\em Proceedings of the Twenty-eighth Annual ACM Symposium on
  Theory of Computing}, STOC '96, pages 571--575, New York, NY, USA, 1996. ACM.
\newblock \href {http://dx.doi.org/10.1145/237814.238005}
  {\path{doi:10.1145/237814.238005}}.

\bibitem{Walderveen:2013}
Freek van Walderveen, Norbert Zeh, and Lars Arge.
\newblock Multiway simple cycle separators and {I/O}-efficient algorithms for
  planar graphs.
\newblock In {\em Proceedings of the Twenty-Fourth Annual {ACM-SIAM} Symposium
  on Discrete Algorithms, {SODA} 2013, New Orleans, Louisiana, USA, January
  6-8, 2013}, pages 901--918, 2013.
\newblock \href {http://dx.doi.org/10.1137/1.9781611973105.65}
  {\path{doi:10.1137/1.9781611973105.65}}.

\end{thebibliography}

\newpage
\appendix

\section{Omitted Proofs}\label{a:omitted}

\uniquematching*

\begin{proof}
The %legacy
algorithm by Gabow et al.~\cite{Gabow:2001} for this problem runs
  in $O(n\log{n})$ time.
  The algorithm has two bottlenecks and otherwise runs in $O(n)$ time.
\begin{enumerate}
\item Maintaining the set of bridges of $G$ under edge deletions.
\item Maintaining the sizes of connected components of $G$ under edge deletions.
  Specifically, one has to be able to query the size of a component containing given $v\in V$ in $O(1)$ time.
\end{enumerate}
  Clearly, the data structure of Theorem~\ref{thm:2edge_ds} can be used
  to remove the former bottleneck.
  The latter bottleneck can be dealt with by extending the data structure for decremental connectivity in planar graphs due to Łącki and Sankowski \cite{Lacki:2015}.
This data structure computes a $r$-division $\rdiv$ of the input graph, and based on it defines a skeleton graph. Roughly speaking, the skeleton graph is defined as follows. We say that a connected component is interesting if it contains a boundary vertex. Thus, each connected component is either interesting or fully contained within one piece of the $r$-division (in which case it is handled with a recursive call).

The skeleton graph represents all interesting connected components of the graph. It has vertices of two types, namely it contains all boundary vertices of the $r$-division and, for each interesting component~$C$ and each piece containing vertices of $C$, one auxiliary vertex representing the intersection of $C$ and the piece. Such an auxiliary vertex may correspond to multiple vertices in the entire graph.

The skeleton graph has $O(n / \sqrt{r})$ vertices, and for each vertex the data structure explicitly maintains the identifier of its connected component. In order to extend the data structure to maintain the sizes of the components, it suffices to maintain, for each auxiliary vertex, the number of vertices in the entire graph, that it corresponds to. From the algorithm, it follows that this information can be updated without impacting the overall running time. 
\end{proof}

\begin{restatable}{theorem}{vertds}
    \label{thm:2vert_ds}
  Let $G=(V,E)$ be a planar graph and let $n=|V|$.
  There exists a deterministic data structure that maintains $G$ subject to edge deletions and can answer $2$-vertex connectivity queries in $O(1)$ time.
  Its total update time is $O(n \log n)$.
\end{restatable}

\begin{proof}
  The only bottleneck of the data structure of \cite{Giammarresi:96} is the following
  subproblem (otherwise the total cost of the data structure is $O(n\log{n})$).
  Suppose we delete an edge $e$ separating the faces $f_l,f_r$, $f_l\neq f_r$.
  Denote by $C(f)$ the cycle bounding the face $f$.
  We want to find the vertices of $C(f_l)\cap C(f_r)$ in order they appear
  on this cycles (the order is the same for both faces up to reversal).
  In the data structure of \cite{Giammarresi:96}, the cycles $C(f_l)$ are
  represented as doubly linked lists and thus they can be maintained in amortized
  constant time under edge deletions (which correspond to face merges).
  The set $C(f_l)\cap C(f_r)$ is computed by iterating through the shorter
  bounding cycle (say $C(f_l)$) and checking for each $v\in C(f_l)$ whether
  $v$ is adjacent with $f_r$.
  This, in turn, is accomplished by storing for each vertex $v\in V$ the set
  of neighboring faces in a balanced binary search tree.
  Consequently $C(f_l)\cap C(f_r)$ is computed in $O(|C(f_l)|\log{n})$ time.
  As we always iterate through the smaller of the cycles which are subsequently
  joined, this gives us $O(n\log^2{n})$ total time for any sequence of edge deletions.

  We now show how the step of computing $C(f_l)\cap C(f_r)$ in order can be sped
  up to $O(|C(f_l)|)$.
  This will make the whole data structure handle any sequence of updates in $O(n\log{n})$
  total time.

  To proceed, we need the notion of a \emph{face-vertex} graph of $G$, denoted by $\fv{G}$.
  This is a plane embedded graph, which is constructed as follows.
  First, embed a single vertex inside every face of $G$, thus obtaining a set of vertices $F$.
  The vertex set of $\fv{G}$ is $V(G) \cup F$.
  We call each element of $V(G)$ a \emph{v-vertex} and each element of $F$ an \emph{f-vertex}.
  Now, consider each face of $G$ one by one.
  For a face $f$ let $v_1, \ldots, v_k$ be the sequence of vertices on the boundary of $f$ (note that we may have $v_i = v_j$ for $i \neq j$).
  Then for each $1 \leq i \leq k$, $\fv{G}$ has a single edge connecting vertex $v_i$ with the vertex embedded inside the face $f$.
  No other edges are added to $\fv{G}$.
  In particular, every edge of $\fv{G}$ connects an f-vertex and a v-vertex, so $\fv{G}$ is bipartite.
  Also, we may have multiple edges between two vertices of $\fv{G}$, if the boundary of some face of $G$ goes through a vertex multiple times.

  We build the data structure $D(H)$ of Theorem~\ref{thm:main} for
  the graph $H=\fv{G}\cup\dual{G}$.
  Clearly, this graph is planar.
  The deletions of edge of $G$ are reflected in $H$ by contractions of edges
  connecting the faces.
  Note that for each $v\in C(f_l)\cap C(f_r)$ there exist edges $vf_l$
  and $vf_r$ in $H$.
  Thus, after merging $f_l$ and $f_r$ into a face $f$, we will have at least two edges
  $vf$ in $H$ that have not been previously parallel.
  Hence, some parallelism $e\to e'$, where $e=vf$ will be reported.
  As a result, we obtain the set $C(f_l)\cap C(f_r)$ from the set of 
  parallel edges reported by $D(H)$ after the last contraction.
  Note that $D(H)$ might also report some parallel edges connecting
  two face-vertices of $H$; such edges are ignored.

  The total time for obtaining all the sets $C(f_l)\cap C(f_r)$ is linear,
  by Theorem~\ref{thm:main}.
  However, recall that the data structure of \cite{Giammarresi:96} requires
  the elements $C(f_l)\cap C(f_r)$ in order of their occurrences on the cycle bounding $f_l$
  and unfortunately $D(H)$ does not give us this order.
  That is why we also need to traverse $C(f_l)$ to obtain the order of $C(f_l)\cap C(f_r)$.
  Hence, $O(|C(f_l)|)$ additional time is needed.
\end{proof}

\begin{restatable}{theorem}{tedgeds}
\label{thm:3edge_ds}
  Let $G=(V,E)$ be a planar graph and let $n=|V|$.
  There exists a deterministic data structure that maintains $G$ subject to edge deletions and can answer $3$-edge connectivity queries in $O(1)$ time.
  Its total update time is $O(n \log n)$.
\end{restatable}

\begin{proof}
The data structure of \cite{Giammarresi:96} maintains explicitly
  the so-called \emph{cactus tree} which succinctly
  describes the structure of $2$-edge-cuts in $G$.
  The vertices of a cactus trees $T$ are the $3$-edge-components
  of $G$ and the edge set of $T$ consist of edge-disjoint
  simple cycles.

The core problem of the update procedure is deleting
  an edge contained in a $3$-edge-connected component $C$.
When such edge $e$ of $G$ is deleted, a vertex of the cactus tree
representing $C$ is possibly split
and some cycles of $T$ are updated;
certain pairs of cycles of the cactus tree
are merged, whereas some cycles get extended by a single edge.

Although not stated explicitly, the total cost of maintaining
the cactus tree $T$ in \cite{Giammarresi:96}, once we know
which pairs of cycles should be merged and which cycles should
be extended after a deletion,
is $O(n\log{n})$, as the total number of updates to $T$
is in fact linear for any sequence of edge deletions
(in each such case, the number of vertices of $T$ grows by at least $1$).

The most computationally demanding part of the procedure updating $T$
  in \cite{Giammarresi:96} is deciding which pairs of cycles of $T$
  should be merged and which cycles should be extended:
  it might take as much as $\Theta(n\log^2{n})$ time for the entire sequence
  of edge updates to $G$.
  This problem is reduced in \cite{Giammarresi:96} to the following.
  Let the deleted edge $e$ separate two faces $f_l$ and $f_r$ of $G$ ($f_l\neq f_r$).
  We need to find the set of faces $f$ of $G$ such that
  $f\notin\{f_l,f_r\}$ and $f$ is neighboring with both $f_l$ and $f_r$
  before the edge deletion, or, in other words, $\{f_lf,f_rf\}\subseteq E(\dual{G})$.

  Note that after contraction $\dual{e}$ in $\dual{G}$ (which identifies the faces
  $f_l$ and $f_r$), all such pairs
  of edges constitute pairs of parallel edges of $\dual{G}$ that have
  not been previously parallel.
  Recall that if $\dual{G}$ was maintained using
  the data structure of Theorem~\ref{thm:main},
  a directed parallelism $e_1\to e_2$ or $e_2\to e_1$, where $e_1=f_lf \land e_2=f_rf$
  would have been reported immediately after contracting $\dual{e}$.

  Consequently, we can solve this subproblem in $O(n)$ total time
  by maintaining the graph $\dual{G}$ under contractions using
 the data structure of Theorem~\ref{thm:main}.
  As all other subproblems of the update procedure of \cite{Giammarresi:96} are
  solved in $O(n\log{n})$ total time, the theorem follows.
\end{proof}

\begin{lemma}\label{lem:maximal-edge-connected}
  Let $k\geq 2$.
  Suppose there exists a data structure $D_k$ maintaining a planar
  graph $H$ under edge contractions and reporting edges of $H$ participating
in some cycle of length $i$, $i\leq k$, in an online manner.
  Denote by $O(f_k(m))$ the total time needed by $D_k$ to execute any sequence of contractions on a graph on $m$ edges.
  Then, there exists an algorithm computing the maximal
  $(k+1)$-edge-connected subgraphs of a planar graph
  $G$ in $O(f_k(m))$ time.
\end{lemma}
\begin{proof}
We build a data structure $D_k$ for $\dual{G}$.
Recall that each simple cycle of length $i\leq k$ in $\dual{G}$
corresponds to a simple edge-cut of size $i$ in $G$.

Clearly, if $G$ does not have edge-cuts of size no more than $k$,
it is $(k+1)$-edge connected.
Suppose some edge $e$ participates in some edge-cut of size $i$
in $G$.
Then $e$ participates in some edge-cut of size no more than $i$
in every subgraph of $G$ containing $e$.
As a result, any subgraph of $G$ containing $e$ is not $(k+1)$-edge-connected.
  We may thus safely discard $e$: the maximal $(k+1)$-edge-connected
  subgraphs of $G$ and $G-e$ are the same.

  The above observation leads to a simple algorithm for computing maximal $(k+1)$-edge-connected
  subgraphs.
  We maintain a queue $Q$ of edges $e$ that participate in some edge-cut of $G$ of size no more than $k$.
  As long as $Q\neq\emptyset$, we extract some $e$ of $Q$ and
  remove it from $G$.
  Recall that removal of an edge $e$ in $G$ corresponds
  to a contraction of $\dual{e}$ in $\dual{G}$.
  We thus issue a contraction to $D_k$,
  after which new edges can be marked as participating in small
  cycles of $\dual{G}$.
  These newly marked edges in $\dual{G}$ are added to $Q$.

  The connected components of the subgraph of $G$ induced by the
  edges that were not removed constitute the maximal $(k+1)$-edge-connected
  subgraphs of $G$.
\end{proof}

\maxedgecon*
\begin{proof}
Recall that, by Theorem~\ref{thm:main}, we can maintain
any planar graph $H$ under edge contractions in linear total time,
so that the edges participating in $2$-cycles, i.e., parallel edges,
are reported in an online fashion.
To finish the proof, we apply Lemma~\ref{lem:maximal-edge-connected}.
\end{proof}

\subsection{Simple and Bounded-Degree}\label{a:degree}
If $G$ is not initially simple, we first report and discard the self-loops.
For each group of pairwise parallel edges $\{e_1,\ldots,e_q\}$
we report directed parallelisms $e_i\to e_1$ for each $i=2,\ldots,q$
and remove all edges $e_2,\ldots,e_q$ from $G$.

Suppose now that $G$ is simple.
We first compute the embedding of $G$ in linear time \cite{Hopcroft:74}.
We build a graph $H$ obtained from $G$
by replacing each vertex $v$ of $G$ of degree at least $4$
with an undirected cycle $C_v$ of length $\deg(v)$ in a standard way
(but consistently with the embedding of $G$).
After this step, all the vertices of $H$ have degree no more than $3$.
Moreover, $H$ is planar and has a linear number of edges and vertices.
If we now build a data structure for maintaining $H$ under contractions,
we may first contract all the edges constituting the cycles $C_v$ first,
but without reporting the self-loops and parallelisms involving these
edges.
Note that after these contractions the graph $H$ in fact becomes equal to $G$.
Clearly, the size of $H$ is linear in the size of $G$.

\subsection{Supporting Edge Weights}\label{sec:edgeweights}
In this section we show how to modify the data structure of
Section~\ref{sec:loopdetection} so that, given a weight function
$\ell:E_0\to\mathbb{R}$, for each
reported directed parallelism $\alpha(Y_{3-i})\to\alpha(Y_i)$ (see Section~\ref{sec:interface})
we have $\ell(\alpha(Y_{i}))\leq\ell(\alpha(Y_{3-i}))$.

We maintain an array $\delta$ defined as follows. 
Let $Y$ be a group of parallel edges represented
by a tree $T\in\parf$.
Then, $\delta[\alpha(Y)]$ is equal to an edge $e\in T$
such that $\ell(e)$ is minimum.
Initially, for each $e\in E_0$, we have $\delta(e)=e$.
To maintain the invariant posed on $\delta$ throughout
any sequence of contractions, we do the following.

Suppose the data structure of Section~\ref{sec:loopdetection} reports
a parallelism $\alpha(Y_1)\to\alpha(Y_2)$.
Then, if $\delta[\alpha(Y_1)]\leq \delta[\alpha(Y_2)]$,
the weight supporting layer reports $\delta[\alpha(Y_2)]\to\delta[\alpha(Y_1)]$
instead
and sets $\delta[\alpha(Y_2)]=\delta[\alpha(Y_1)]$.
On the other hand, when $\delta[\alpha(Y_1)]>\delta[\alpha(Y_2)]$,
we only report $\delta[\alpha(Y_1)]\to\delta[\alpha(Y_2)]$ instead.

\newpage
\section{Pseudocode of Linear Time Algorithms for 5-coloring and MST}\label{a:pseudocode}
%\begin{figure}
  \vspace{-4mm}
\SetProcNameSty{texttt}
\begin{function*}[h]
\DontPrintSemicolon
$(s,P,L)\leftarrow\dscontr(e)$\;
  \For{$w\in \{s\}\cup \bigcup\{\dsverts(e_1):(e_1\to e_2)\in P\}$} {
  \If{$w\notin Q$ {\bf and} $\dsdeg(w)\leq 5$}{
    $Q\leftarrow Q\cup \{w\}$
  }
}
\Return{$s$}
\caption{contract-and-update($e$)}
\end{function*}

  \vspace{-4mm}
\begin{algorithm*}[H]
\SetAlgoFuncName{Algorithm}
  \SetAlgoRefName{}
\DontPrintSemicolon
\SetKwInOut{Input}{input}
  \SetKwInOut{Output}{output}
  \Input{A simple connected planar graph $G=(V,E)$ and a function $\ell:E\to\mathbb{R}$.}
\Output{A minimum spanning tree of $G$.}
  $\dsinit(G)$. Use $\ell$ to report directed parallelism, so that each time $e'\to e$ is reported,
  we have $\ell(e')\geq \ell(e)$.\;
  $Q\leftarrow \{v\in V: \dsdeg(v)\leq 5\}$\;
  $T\leftarrow \emptyset$\;
  \While{$Q\neq\emptyset$}{
    $u\leftarrow\textup{any element of } Q$\;
    $Q\leftarrow Q\setminus\{u\}$\;
    \If{$\dsdeg(u)\geq 1$}{
      $e\leftarrow \textup{an edge such that }(v,e)\in \dsnei(u)\textup{ and }\ell(e)\textup{ is minimal}$\;
      $T\leftarrow T\cup\{e\}$\;
      $\texttt{contract-and-update}(e)$\;
    }
  }
  \Return{$T$}
  \caption{MST of a planar graph}
\end{algorithm*}

\begin{procedure*}[H]
\DontPrintSemicolon
  \If{$Q=\emptyset$}{
    \Return
  }
  $u\leftarrow\textup{any element of } Q$\;
  $Q\leftarrow Q\setminus\{u\}$\;
  $Z\leftarrow \dsnei(u)$\;
  \If{$\dsdeg(u)\geq 1$ {\bf and} $\dsdeg(u)\leq 4$}{
    $v\leftarrow\textup{any vertex of } Z$\;
    $s\leftarrow\texttt{contract-and-update}(\dsedge(u,v))$\;
    $\texttt{color()}$\;
    $C[v]\leftarrow C[s]$\;
  }
  \ElseIf{$\dsdeg(u)=5$}{
    $x,y\leftarrow\textup{any two vertices of } Z\textup{ such that }\dsedge(x,y)=\nil$\;
    $s'\leftarrow\texttt{contract-and-update}(\dsedge(u,x))$\;
    $s\leftarrow\texttt{contract-and-update}(\dsedge(s_1,y))$\;
    $\texttt{color()}$\;
    $C[x]\leftarrow C[s]$\;
    $C[y]\leftarrow C[s]$\;
  }
  $C[u]\leftarrow\textup{ any color of }\left(\{1,2,3,4,5\}\setminus \{C[w]:w\in Z\}\right)$\;
\caption{color()}
\end{procedure*}

\begin{algorithm*}[H]
\SetAlgoFuncName{Algorithm}
  \SetAlgoRefName{}
\DontPrintSemicolon
\SetKwInOut{Input}{input}
  \SetKwInOut{Output}{output}
\Input{A simple connected planar graph $G=(V,E)$.}
\Output{A 5-coloring $C$ of $G$.}
  $\dsinit(G)$\;
  $Q\leftarrow \{v\in V: \dsdeg(v)\leq 5\}$\;
  $C\leftarrow \textup{ an array indexed with vertices with values from }\{1,2,3,4,5\}$\;
  $\texttt{color()}$\;
  \Return{$C$}
  \caption{5-coloring of a planar graph}
\end{algorithm*}

\end{document}